\documentclass[10pt, a4paper]{article}

\usepackage{hyperref}
\usepackage{amsfonts,amssymb,mathtools}
\usepackage{amssymb}
\usepackage{amsmath}
\usepackage{graphicx}
\usepackage[english]{babel}
\usepackage[latin1]{inputenc}
\usepackage{verbatim}
\usepackage{enumerate}
\usepackage{tikz}
\usetikzlibrary{backgrounds,calc}
\usepackage{float}
\usepackage{color}
\usepackage{acronym}
\usepackage{algorithm}
\usepackage{algorithmic}
\newfloat{algorithm}{H}{loa}
\floatname{algorithm}{Algorithm}
\newcounter{alg}
\usepackage{multirow}
\newtheorem{theorem}{Theorem}

\newtheorem{corollary}[theorem]{Corollary}
\newtheorem{lemma}[theorem]{Lemma}
\newtheorem{claim}[theorem]{Claim}

\newtheorem{definition}{Definition}

\newtheorem{observation}[theorem]{Observation}
                      {}
\newenvironment{prooff}[1]{\begin{trivlist}
\item[\hspace{\labelsep}{\bf\noindent Proof of #1: }]
}{{\hspace*{\fill}{\eod}}\end{trivlist}}
                      {}

\def\squarebox#1{\hbox to #1{\hfill\vbox to #1{\vfill}}}

\def\eod{\vrule height 6pt width 5pt depth 0pt}
\newenvironment{proof}{\noindent {\bf Proof:} \hspace{.677em}}{\hspace*{\fill}{\eod}}

\allowdisplaybreaks

\begin{document}

\title{Improving approximate pure Nash equilibria in congestion games}

\author{Alexander Skopalik\thanks{Mathematics of Operations Research, University of Twente, Netherlands, E-mail:a.skopalik@utwente.nl}\and
Vipin Ravindran Vijayalakshmi\thanks{Chair of Management Science, RWTH Aachen, Germany. E-mail:vipin.rv@oms.rwth-aachen.de. This work is supported by the German research council (DFG) Research Training Group 2236 UnRAVeL.}}

\date{}
\maketitle


\begin{abstract}
Congestion games constitute an important class of games to model resource allocation by different users. As computing an exact~\cite{Fabrikant:2004:CPN:1007352.1007445} or even an approximate~\cite{Skopalik:2008:IPN:1374376.1374428} pure Nash equilibrium 
is in general PLS-complete, Caragiannis et al.~\cite{Caragiannis2011} present a polynomial-time algorithm that computes a (2 + $\epsilon$)-approximate pure Nash equilibria for games with linear cost functions and further results for polynomial cost functions. We show that this factor can be improved to $(1.61+\epsilon)$ and further improved results for polynomial cost functions, by a seemingly simple modification to their algorithm by allowing for the cost functions used during the best response dynamics be different from the overall objective function. Interestingly, our modification to the algorithm also extends to efficiently computing improved approximate pure Nash equilibria in games with arbitrary non-decreasing resource cost functions. Additionally, our analysis exhibits an interesting method to optimally compute universal load dependent taxes and using linear programming duality prove tight bounds on PoA under universal taxation, e.g, 2.012 for linear congestion games and further results for polynomial cost functions. Although our approach yield weaker results than that in Bil\`{o} and Vinci~\cite{Bilo}, we remark that our cost functions are locally computable and in contrast to~\cite{Bilo} are independent of the actual instance of the game.

\end{abstract}

\section{Introduction}
   Congestion games constitute an important class of games that succinctly represents a game theoretic model for resource allocation among non-cooperative users. A canonical example for this is the road transportation network, where the time needed to commute is a function on the total amount of traffic in the network (see e.g.~\cite{TransportUS}). A congestion game is a cost minimization game defined by a set of resources $E$, a set of $n$ players with strategies $S_1,\dots ,S_n \subseteq 2^E$, and for each resource $e \in E$, a cost function $f_e : \mathbb{N} \mapsto \mathbb{R_+}$. Congestion games were first introduced by Rosenthal~\cite{Rosenthal1973} and using a potential function argument proved that it belongs to a class of games in which a pure Nash equilibrium always exists, i.e., the game always consists of a self-emerging solution in which no user is able to improve by unilaterally deviating. For a strategy profile $s \in S_1 \times \cdots \times S_n$, the cost of a player $u \in \mathcal{N}$ is defined as $c_u(s) := \sum_{e \in s_u}f_e(n_e(s))$, where $n_e(s)$ denotes the number of players on the resource $e$ in the state $s$. The potential of the game in the state $s$ is defined as $\phi(s) := \sum_{e \in E}\sum_{i=1}^{n_e(s)}f_e(i)$. 
    
   \subsubsection*{Convergence to pure Nash equilibria} Fabrikant et al.~\cite{Fabrikant:2004:CPN:1007352.1007445} show that computing a pure Nash equilibrium in both symmetric and asymmetric  congestion games\footnote{A congestion game is called symmetric if $S_i=S_j$ for all $i,j\in \mathcal{N}$. Otherwise it is called asymmetric.} is PLS-complete. They show that regardless of the order in which local search is performed, there are initial states from where it could take exponential number of the steps before the game converges to a pure Nash equilibrium. Also, they show PLS-completeness for network congestion games with asymmetric strategy spaces. As a positive result, Fabrikant et al.~\cite{Fabrikant:2004:CPN:1007352.1007445} present a polynomial time algorithm to compute a pure Nash equilibria in certain restricted strategy spaces e.g. symmetric network congestion games. Ackermann et al.~\cite{Ackermann:2008:ICS:1455248.1455249} show that network congestion games with linear cost functions are PLS-complete. However, if the set of strategies of each player consists of the bases of a matroid over the set of resources, then they show that the lengths of all best response sequences are polynomially bounded in the number of players and resources. This alludes for studying {\em approximate} pure Nash equilibria in congestion games.
    
    To our knowledge, the concept of $\alpha$-approximate equilibria\footnote{Here we refer to the multiplicative notion of approximation. There is also a additive variant which is often denoted by $\epsilon$-Nash.} was introduced by Roughgarden and Tardos~\cite{Tim_Eva} in the context of non-atomic selfish routing games. An $\alpha$-approximate pure Nash equilibrium is a state in which none of the users can unilaterally deviate to improve by a factor of at least $\alpha$. Orlin et al.~\cite{Orlin:2004:ALS:982792.982880} show that every local search problem in PLS admits a fully polynomial time $\epsilon$-approximation scheme. Although their approach can be applied to congestion games, this does not yield an approximate pure Nash equilibrium, but rather only an approximate local optimum of the potential function. In case of congestion games, Skopalik and V{\"o}cking~\cite{Skopalik:2008:IPN:1374376.1374428} show that in general for arbitrary cost functions, finding a $\alpha$-approximate pure Nash equilibrium is PLS-complete, for any $\alpha>1$. However, for polynomial cost function (with non-negative coefficients) of maximum degree $d$, Caragiannis et al.~\cite{Caragiannis2011} present an approximation algorithm. They present a polynomial-time algorithm that computes (2 + $\epsilon$)-approximate pure Nash equilibria for games with linear cost functions and an approximation guarantee of $d^{O(d)}$ for polynomial cost functions of maximum degree $d$. Interestingly, they use the convergence of subsets\footnote{This subset is carefully chosen such that convergences in polynomial time is guaranteed.} of players to a $(1+\epsilon)$-approximate Nash equilibrium (of that subset) as a subroutine to generate a state which is an approximation of the minimal  potential function value (of that subset), e.g., $2\cdot \textsc{opt}$ for linear congestion games. This approximation factor of the minimal potential then essentially turns into the approximation factor of the approximate equilibrium.
    Feldotto et al.~\cite{Feldotto} using a \textit{path-cycle decomposition} technique bound this approximation factor of the potential for arbitrary cost functions. 
    
    	\subsubsection*{Load Dependent Universal Taxes} The last 20 years saw a significant amount of literature investigating the bound on the \textit{price of anarchy} (\textsc{PoA})~\cite{Koutsoupias2009} for various non-cooperative games and several attempts to improve the inefficiency of these self-emerging solutions. One of the many approaches used to improve the \textsc{PoA} is the introduction of taxes~\cite{Caragiannis, Fotakis, Fotakis2007}. For a set of resources $E$, the load dependent tax function $t$, is the excess cost incurred by the user on a resource $e \in E$ with cost $f(x)$, e.g., $f'(x) = f(x) + t(x)$. Meyers and Schulz~\cite{Meyers:2012:CWM:2305836.2305842} study the complexity of computing an optimal solution in a congestion game and prove NP-hardness. Makarychev and Sviridenko~\cite{Makarychev} give the best known approximation algorithm using randomized rounding on a natural feasibility LP with approximation factor
$\mathcal{B}_{d+1}$ which is the $d+1^{\text{th}}$ Bell number, where $d$ is the maximum degree of the polynomial cost function. 
    Interestingly, the same was later achieved using load dependent taxes by Bil\`{o} and Vinci~\cite{Bilo}, where they apply the \textit{primal-dual} method~\cite{Bilo:2018:UTB:3214123.3214125} to upper bound the \textsc{PoA} under refundable taxation in congestion games. They determine a load specific taxation to show that the \textsc{PoA} is at most $[O({d}/{\log d})]^{d+1}$ with respect to $\epsilon$-approximate equilibrium under refundable taxation. However, we remark that the load dependent taxes computed in~\cite{Bilo} aren't universal, i.e., they are sensitive to the instance of the game.
    
    \subsubsection*{Our Contribution} In this paper we improve the approximation guarantee achieved in the computation of approximate pure Nash equilibrium with the algorithm in Caragiannis et al.~\cite{Caragiannis2011}, using a linear programming approach which generalizes the smoothness condition in Roughgarden~\cite{Roughgarden:2015:IRP:2841330.2806883}, to modify the cost functions that users experience in the algorithm. In Section~\ref{CG_stretch} we present an adaptation of the algorithm in Caragiannis et al.~\cite{Caragiannis2011}. Although we only make a seemingly simple modification to their algorithm, we would like to remark that the analysis is significantly involved and does not follow immediately from~\cite{Caragiannis2011}, since the sub-game induced by the algorithm with the modified costs is not a potential game anymore. Table~\ref{table:apxPNE} lists the results for resource cost function that are bounded degree polynomials of maximum degree $d$.
    
    
  \begin{table}[H]
	\centering
	\begin{tabular}{ |c|c|c| }	
		\hline
		\text{~d~} & \text{Previous Approx.}~\cite{Caragiannis2011, Feldotto}&$\textbf{Our Approx.}$ $\rho_d + \epsilon$\\
		\hline
		$1$ &  $2+\epsilon$ & $\text{1.61}+\epsilon$\\
		$2$ &  $6+\epsilon$ & $\text{3.35}+\epsilon$\\
		$3$ &  $20+\epsilon$ & $\text{8.60}+\epsilon$\\
		$4$ &  $111+\epsilon$ & $\text{27.46}+\epsilon$\\
		$5$ &  $571+\epsilon$ & $\text{98.14}+\epsilon$\\
		\hline
	\end{tabular}
	\caption{Approximate pure Nash equilibria of congestion games with polynomial cost functions of degree at most $d$.}
	\label{table:apxPNE}
    \end{table}
    Our main contribution in this paper is 
    presented as Theorem~\ref{main_thm}, where the factor $\rho_d$ is listed in Table~\ref{table:apxPNE}.
    
    \begin{theorem}\label{main_thm}
		For every  $\epsilon > 0$, the algorithm computes a $(\rho_d + \epsilon$)-approximate equilibrium for every congestion game with non-decreasing cost functions that are polynomials of maximum degree $d$
	  in a number of steps which is polynomial in the number of players, $\rho_d$ and $1/\epsilon$.
	\end{theorem}
	Our approach also yields a simple and distributed method to compute load dependent universal taxes that improves the inefficiency of equilibria in congestion games. We remark that the taxes we consider in Section~\ref{section_extension} are refundable and do not contribute to the overall cost of the game. Table~\ref{table:locality} lists our results for price of anarchy (PoA) under refundable taxation for resource cost functions that are bounded degree polynomials.
	Bil\`{o} and Vinci~\cite{Bilo} present an algorithm to compute load dependent taxes that improve the price of anarchy e.g., for linear congestion games from 2.5 to 2. Although our methods yield slightly weaker results, 
	our cost functions are locally computable and in contrast to \cite{Bilo} are independent of the actual instance of the game. Furthermore, using linear programming duality we derive a reduction to a selfish scheduling game on identical machines, which implies a matching lower bound on the approximation factor. We would like to remark that our results for PoA were achieved independently of that in Paccagnan et al.~\cite{chandan1} by a very similar technique.
	   \begin{table}[H]
	\centering
	\begin{tabular}{ |c|c|c|c| }	
		\hline
		\multirow{2}{*}{$\text{~d~}$} & $\text{\text{PoA without Taxes}}$ & \text{Optimal Taxes}& $\textbf{Universal Taxes } \Psi_d$\\
		 &\text{Aland et al.~\cite{Games2011}}& \text{Bil\`{o} and Vinci~\cite{Bilo}}& $\text{Local Search w.r.t}~\zeta_{\textsc{sc}}$\\
		\hline
		$1$ & $2.5$ & $2$ & $\text{2.012}$\\
		$2$ & $9.583$ & $5$ &$\text{5.10}$\\
		$3$ & $41.54$ & $15$ &$\text{15.56}$\\
		$4$ & $267.6$ & $52$& $\text{65.12}$\\
		$5$ & $1414$ & $203$& $\text{641.32}$\\
		\hline
	\end{tabular}
	\caption{PoA under taxation in congestion games with polynomial cost functions of degree at most $d$.}
	\label{table:locality}
  \end{table}
  

\section{Definitions and Preliminaries}\label{def_prelim}
 
    A strategic game denoted by the tuple $\left(\mathcal{N},\left(S_u\right)_{u\in \mathcal{N}},  \left(c_u\right)_{u\in \mathcal{N}}\right)$ consists of a finite set of players $\mathcal{N}$ and for each player $u \in \mathcal{N}$, a finite set of strategies $S_u$ and a cost function $c_u: S \rightarrow \mathbb{R_+}$ mapping a state $s  \in S := S_1 \times S_2 \times\cdots \times S_N$ to the cost of player $u\in \mathcal{N}$. A congestion game is a strategic game that succinctly represents a decentralized resource allocation problem involving selfish users.

    \label{CG_p}
     A congestion game denoted by $G=\left(\mathcal{N},E,\left(S_u\right)_{u\in N}, \left(f_e\right)_{e\in E}\right)$ consists of a set of N players, $\mathcal{N} = \{1,2,\ldots,N\}$, who compete over a set of   resources, $E = \{e_1, e_2,\ldots, e_m\}$. 
     Each player $u \in \mathcal{N}$ has a set of strategies denoted by $S_u \subseteq 2^E$. Each resource $e \in E$ has a non-negative and non-decreasing cost function $f_e:\mathbb{N} \mapsto \mathbb{R_+}$ associated with it. Let $n_e(s)$ denote the number of players on a resource $e\in E$ in the state $s$, then the cost contributed by a resource $e \in E$ to each player using it is denoted by $f_e(n_e(s))$. 
     Therefore, the cost of a player $u \in \mathcal{N}$ in a state $s=(s_1,\ldots,s_N)$ of the game is given by $c_u(s)=\sum_{e \in E:e \in s_u}f_e(n_e(s))$. For a state $s$, $c_u(s_u', s_{-u})$ denotes the cost of player $u$, when only $u$ deviates.
     
     A state $s \in S$ is a pure Nash equilibrium (PNE), if there exists no player who could deviate to another strategy and decrease their cost, i.e., $ \forall u \in \mathcal{N}$ and $\forall s_u^\prime \in S_u$,~$c_u(s) \le  c_u(s_u^\prime, s_{-u})$.
     A weaker notion of PNE is the $\alpha$-approximate pure Nash equilibrium for $\alpha \ge 1$, which is a state  $s$ in which no player has an improvement that decreases their cost by a factor of at least $\alpha$, i.e., $\forall u\in \mathcal{N}$ and $\forall s_u^\prime \in S_u,~ \alpha\cdot c_u(s'_u,s_{-u}) \ge   c_u(s)$. 
 For congestion games the exact potential function $\phi(s) =\sum_{e \in E}\phi_e(n_e(s)) = \sum_{e \in E}\sum_{i=1}^{ n_e(s)}f_e(i)$, guarantees the existence of a PNE by proving that every sequence of unilateral improving strategies converges to a PNE. 
	We denote social or global cost of a state $s$ as $c(s) = \sum_{u\in \mathcal{N}} c_u(s)$ and the state that minimizes social cost is called the optimal, i.e., $s^* = \arg \min_{s \in S}c(s)$. The inefficiency of equilibria is measured using the price of anarchy (\textsc{PoA})~\cite{Koutsoupias2009}, which is the worst case ratio between the social cost of an equilibrium and the social optimum.
	
	A local optimum is a state $s$ in which there is no player $u \in \mathcal{N}$ with an alternative strategy $s'_u$ such that, $c(s'_u,s_{-u})< c(s)$ and an $\alpha$-approximate local optimum is a state $s$ in which there is no player $u$ who has an $\alpha$-move with a strategy $s'_u$ such that $\alpha\cdot c(s'_u,s_{-u})<   c(s)$. Let us remark that there is an interesting connection between a local optimum and a PNE. A PNE is a local optimum of the potential function $\phi$ and similarly, a local optimum is a Nash equilibrium of a game in which we change the resource cost functions from $f(x)$ to the marginal contribution to social cost, e.g., to $f'(x) = x f(x) - (x-1) f(x-1)$. Analogous to the PoA, the {\em stretch} of a congestion game is the worst case ratio between the value of the potential function at an equilibrium and the potential minimizer~\cite{Caragiannis2011}.
	
	\subsection{Revisiting $(\lambda, \mu)$-smoothness}\label{CG}

    After a long series of papers in which various authors~(e.g. \cite{Christodoulou:2005:PAF:1060590.1060600, Awerbuch:2005:PRU:1060590.1060599, Games2011}) show upper bounds on the price of anarchy, Roughgarden exhibited that most of them essentially used the same technique, which is formalized as $(\lambda, \mu)$-smoothness~\cite{Roughgarden:2015:IRP:2841330.2806883}. A game is called $(\lambda, \mu)$-smooth, if for every pair of outcomes $s$, $s^*$, it holds that, \[\sum_{u \in \mathcal{N}}c_u(s_u^*, s_{-u}) \le \lambda \cdot c(s^*) + \mu \cdot c(s).\tag{1}\]
    The price of anarchy of a $(\lambda, \mu)$-smooth game with $\lambda > 0$ and $\mu < 1$ is then at most $\frac{\lambda}{1-\mu}$. Observe that the original smoothness definition (1) can be extended to allow for an arbitrary objective function $h(s)$ instead of the social cost function $c(s)=\sum_{u \in \mathcal{N}} c_u(s)$.
	
	\begin{definition}
		\label{def:lambdasmooth}
		A  game is $(\lambda, \mu)$-smooth {\em with respect to an objective function $h$}, if for every pair of outcome $s$, $s^*$,
	 		\[\lambda \cdot  h(s^*) \geq   \sum_{u \in \mathcal{N}} c_u(s_u^*,s_{-u}) - \sum_{u \in \mathcal{N}} c_u(s) + (1-\mu) h(s).\]
	\end{definition}
	From the definition above, we restate the central smoothness theorem~\cite{Roughgarden:2015:IRP:2841330.2806883}.  
	\begin{theorem}
	\label{theo:smoothness}
		Given a $(\lambda, \mu)$-smooth game $G$ with $\lambda > 0$, $\mu < 1$, and an objective function h, then for every equilibrium $s$ and the global optimum $s^*$, \[h(s) \le \frac{\lambda}{1 - \mu} h(s^*).\]
	\end{theorem}
	The smoothness framework introduced by Roughgarden~\cite{Roughgarden:2015:IRP:2841330.2806883} also extends to equilibrium concepts such as mixed Nash\footnote{In a mixed Nash equilibrium $(\sigma_1,\ldots,\sigma_N)$ each player $u$ chooses a probability distribution $\sigma_u$ of his set of strategies and $\forall u \in \mathcal{N}$, and $\forall \sigma_i^\prime ,~c_u(s) \le  c_u(s_u^\prime, s_{-u})$, where $c_u$ denotes expected cost of player $u$.} and (coarse\footnote{	A probability distribution $\sigma$ over the set of states $S$ is said to be an coarse correlated equilibrium (CCE) is if $\forall s, s^* \in S, \forall u \in \mathcal{N}, ~ \mathbb{E}_{s \sim \sigma}[c_u(s)] \le  \mathbb{E}_{s \sim \sigma}[c_u(s_u^*, s_{-u})]$. 
	}) correlated  equilibria\footnote{ 
 A probability distribution $\sigma$ over the set of states $S$ is said to be an   correlated equilibrium if for every player $i\in \mathcal{N}$ and every two strategies $s_a,s_b \in S_i$ and every recommendation $s=(s_1,\ldots,s_n)\sim \sigma$, the expected cost for following the recommendation $s_u=s_a$ is not greater that choosing $s_b$ instead. 
  }. The same is true for our variant with respect to an arbitrary objective function $h$. For the sake of completeness a rework can be found in the Appendix~\ref{missing_proofs_apxalgo}. 
	
	From Definition~\ref{def:lambdasmooth} we note the following observation.
\begin{observation}\label{obser1}
Every $(\lambda, \mu)$-smooth game G with $\lambda>0$ and $\mu < 1$, is also $(\frac{\lambda}{1-\mu},0)$-smooth with its cost functions scaled by a factor $\frac{1}{1-\mu}$ .
\end{observation}
	
	Given a strategic game $G=\left(\mathcal{N},\left(S_u\right)_{u\in \mathcal{N}},  \left(c_u\right)_{u\in \mathcal{N}}\right)$, one can determine $\lambda$ and $\mu$ that satisfies the smoothness condition in Definition~\ref{def:lambdasmooth}, for all pairs of solution $s, s^*$. However, since the state space $S$ grows exponentially in the number of players, this would be computationally inefficient. Therefore, we typically have to work with games in which the players' costs and the objective function $h$ can be represented in a succinct way. In congestion games, the players cost and the global objective function are implicitly defined by the resource cost function.  In the following, we allow for an arbitrary, additive objective function $h(s)$, i.e., of the form $h(s) = \sum_{e \in E} h_e(n_e(s))$.	

	We study games in which we change the cost functions $c_u$ experienced by the players. As a consequence of Observation~\ref{obser1} and scaling the cost functions appropriately, we can always ensure that we satisfy the smoothness inequality with $\mu=0$, to conveniently restate the smoothness condition as follows.
	\begin{lemma}
		\label{def:usmooth}
		A congestion game is $(\lambda, 0)$-smooth with respect to an objective function $h(s) = \sum_{e \in E} h_e(n_e(s))$, if for every cost function $f'_e:\mathbb{N} \mapsto \mathbb{R}_+$ and for every $ 0\le n,m \le N,$ \[\lambda \cdot  h_e(m)  \geq   m  f'_e\left(n+1\right) -  n f'_e(n) + h_e(n).\]
	\end{lemma}

	The proof in the Appendix~\ref{missing_proofs_apxalgo} follows from summing the inequality of the lemma with $m=n_e(s^*)$ and $n=n_e(s)$ for two arbitrary solutions $s$ and $s^*$, for all $e \in E$. From now on, we use $f'=(f'_e)_{e\in E}$ whenever we refer to the modified cost functions and denote the players cost by $c_u'(s) = \sum_{e \in s_u} f'_e(n_e(s))$.

\subsubsection*{Strong
smoothness}
	In Section~\ref{CG_stretch} we present an algorithm to compute an approximate pure Nash equilibria with an improved approximation guarantee than that in Caragiannis et al.~\cite{Caragiannis2011} and the proof of which uses the potential function argument for a subset of players $F \subseteq \mathcal{N}$. In particular, it needs the property that the subgame induced by every subset of players from $\mathcal{N}$ is $(\lambda, 0)$-smooth. Unfortunately, Lemma~\ref{def:usmooth} does not guarantee this property. Therefore, we define a stronger notion of $(\lambda, 0)$-smoothness that guarantees that the smoothness condition also holds for an arbitrary subset of players and its induced subgame.
	
	Let us denote by $n_e^F(s)$ the number of players in $F$ that use the resource $e$ in the state $s$.
	\begin{definition}
		\label{def:usmooth_subset}
		A strategic game is  strongly $(\lambda, 0)$-smooth with respect to an objective function $h$ and for some $\lambda > 0$, if for every subset $F \subseteq \mathcal{N}$ and for every $s,s^* \in S$,
		\[ \lambda \cdot  h^F(s^*) \geq   \sum_{u \in F}  c'_u(s_u^*,s_{-u}) - \sum_{u \in F}  c'_u(s) + h^F(s), \] where $h^F(s) := \sum_{e \in E}h_e(n_e(s)) - h_e(n_e^{\mathcal{N}\setminus F}(s)).$
	\end{definition}
	 
	 Now consider an arbitrary subset of players $F \subseteq \mathcal{N}$ and a state $s$. Let us define the potential of this subset as the potential in the subgame induced by these players in $s$, i.e., $\phi^F(s) := \sum_{i=1}^{n_e^F(s)}f_e(i + n_e^{\mathcal{N}\setminus F}(s))$. With slight abuse of notation, 
   we remark that $\phi^F(s)$ and $\phi_F(s)$ are equivalent.  Then, $G_s^F:=(F,E,(S_u)_{u \in F}, (f^F_e)_{e_\in E})$ is the  subgame induced by freezing the remaining players from $\mathcal{N} \setminus F$, with $f^F_e(x):=f_e(x+n_e^{\mathcal{N}\setminus F}(s))$, where $n_e^{\mathcal{N}\setminus F}(s)$ is the number of players outside of $F$ on resource $e$ in the state $s$. 
  Then, the following lemma gives a stronger notion of the $(\lambda, 0)$-smoothness condition.

 	  
    \begin{lemma}
    \label{lem:subset}
    For every congestion game $G$ with cost functions $f'_e:\mathbb{N} \mapsto \mathbb{R}_+$, which is $(\lambda, 0)$-smooth with respect to the potential function $\phi_e$ for every subgame $G^F_s$ induced by an arbitrary subset $F \subseteq \mathcal{N}$, and arbitrary states $s,s^* \in S$, i.e., \[\lambda \cdot \phi_e^F(s^*) - n_e^F(s^*)\cdot f_e'(n_e(s)+1) + n_e^F(s) \cdot f_e'(n_e(s)) \ge \phi_e^F(s),\] 
    is also strongly $(\lambda, 0)$-smooth. 
    \end{lemma}
    
 
 The proof of the lemma is shifted to the Appendix~\ref{missing_proofs_apxalgo}. 
 This subset property is of particular importance for the algorithm we present in Section~\ref{CG_stretch} to compute an approximate equilibrium, but may be of independent interest as well. We are not aware of other approximation algorithms that can guarantee this property as well. We would like to remark that all references to $(\lambda, 0)$-smoothness in Section~\ref{CG_stretch} imply strong $(\lambda, 0)$-smoothness. 


   	\section{Approximate Equilibria in Congestion Games}\label{CG_stretch}
   	In this section we aim at improving the approximation factor of an approximate pure Nash equilibria in congestion games with arbitrary non-decreasing resource cost functions. 
	We extend an  algorithm based on Caragiannis et al.~\cite{Caragiannis2011} to compute an approximate pure Nash equilibrium in congestion games with arbitrary non-decreasing resource cost functions. 
	A key element of this algorithm is the so called stretch of a (sub-) game. This is the worst case ratio of the potential function at an equilibrium and the global minimum of the potential. 
 
	This algorithm  generates a sequence of improving moves  that converges to an approximate pure Nash equilibrium in polynomial number of best-response moves. The idea is to divide the players into blocks based on their costs and hence their prospective ability to drop the potential of the game. In each phase of the algorithm, players of two consecutive blocks are scheduled to make improving moves starting with the blocks of players with high costs. One block only makes $q$-moves, which are improvements by a factor of at least $q$ which is close to $1$. The other block does $p$-moves, where $p$ is slightly larger than the stretch of a $q$-approximate equilibrium and slightly smaller than the final approximation factor.
	
	The key idea here is that blocks first converge to a $q$-approximate equilibrium and thereby generate a state with a stretch of approximately $p$. Later, when players of a block are allowed to do $p$-moves, there is not much potential left to move. In particular, there is no significant influence on players of blocks that moved earlier possible. This finally results in the approximation factor of roughly $p$.
	We modify the algorithm in~\cite{Caragiannis2011} by changing the costs seen by the players during their $q$-moves to be a set of modified cost functions $(f^\prime_e)_{e \in E}$, satisfying smoothness condition of Lemma~\ref{lem:subset} for some constant $\lambda>0$, 
	and this results in a $\lambda(1 + \epsilon)$-approximate pure Nash equilibrium. Note that, $\lambda$ is the stretch with respect to the modified cost functions. For the sake of completeness we present the algorithm as Algorithm~\ref{algo}, but note that only the definition of $\theta(q)$  using $\lambda$, the definition of $p$ in Line 1, and the use of the modified cost functions $(f^\prime_e)_{e \in E}$ in Line~\ref{line11} has been changed. 
    \begin{algorithm*}[ht]
		\caption{Computing a $\lambda (1 + \epsilon)$-approximate pure Nash equilibria in congestion games.}		\label{algo}
		\begin{algorithmic}[1]
			\renewcommand{\algorithmicrequire}{\textbf{Input:}}
			\renewcommand{\algorithmicensure}{\textbf{Output:}}
			\renewcommand{\algorithmicforall}{\textbf{foreach}}
			\REQUIRE Congestion game ${G} = \left(\mathcal{N},E,\left(S_u\right)_{u\in \mathcal{N}}, \left(f_e\right)_{e\in E}\right)$, $f^\prime :=(f^\prime_e)_{e \in E}$ and  $\epsilon>0$.
			\ENSURE  A state of ${G}$ in $\lambda (1 + \epsilon)$-approximate pure Nash equilibrium.
			\STATE Set $q = \left(1 + \frac{1}{N^c}\right)$ , $p = \left(\frac{1}{\theta\left(q\right)} - \frac{1 + q+  2 \lambda}{N^c}\right)^{-1}$, $c = 10 \log\left(\frac{\lambda}{\epsilon}\right)$, $\Delta = \max_{e \in E} \frac{f_e(N)}{f_e(1)}$ and $\theta(q) = \frac{\lambda}{1 + \frac{1-q}{q} N\lambda}$, where  $\lambda:= \min\{\lambda^\prime \in \mathbb{R^+} : \lambda^\prime$ satisfies Lemma~\ref{lem:subset} with respect to $(f^\prime_e)_{e \in E}$\}.
			\FORALL {$u\in \mathcal{N}$}
			\STATE set $\mathit{\ell_u} = c_u\left(\mathcal{BR}_u\left(0\right)\right)$;
			\ENDFOR
			\STATE Set $\mathit{\ell_{min}} = \min_{u\in \mathcal{N}}\mathit{\ell_u}$, $\mathit{\ell_{max} = \max_{u\in \mathcal{N}}\mathit{\ell_u}}$ and $\hat{z} = 1 + \lceil \log_{2\Delta N^{2c+2}}\left(\ell_{max}\slash \ell_{min}\right)\rceil$;
			\STATE Assign players to blocks $B_1, B_2,\cdots ,B_{\hat{z}}$ such that
			\\ $u\in B_i \Leftrightarrow \mathit{\ell_u} \in \left(\mathit{\ell_{max}}\left(2\Delta N^{2c+2}\right)^{-i}, \mathit{\ell_{max}}\left(2\Delta N^{2c+2}\right)^{-i+1}\right]$;
			\FORALL {$u \in N$}
			\STATE set the player $u$ to play the strategy $s_u \leftarrow \mathcal{BR}_u\left(0\right)$;
			\ENDFOR
			\FOR {phase $i \leftarrow 1$ to $\hat{z}-1$ such that $B_i \ne \emptyset$}
			\WHILE {$\exists u \in B_i$ with a $p$-move w.r.t the original cost $f$ or $\exists u \in B_{i+1}$ with a $q$-move w.r.t to modified cost $f^\prime$} \label{line11}
			\STATE $u$ deviates to that best-response strategy $s_u \leftarrow  \mathcal{BR}\left(s_1,\cdots,s_n\right)$.
			\ENDWHILE
			\ENDFOR
		\end{algorithmic}
	\end{algorithm*}


	\subsection{Analysis of the Algorithm}
	\label{analysis_sec}
	We are now ready to prove Theorem~\ref{main_thm}, by restating it as follows. The proof of the theorem follows the proof scheme of Caragiannis et al.~\cite{Caragiannis2011}, which we have to rework to accommodate for our modifications stated above.
    
    	\begin{theorem}
	\label{C_T_1}
		For every constant $\epsilon>0$ and every set of  cost functions $\left(f_e^\prime\right)_{e \in E}$ which are strongly $(\lambda, 0)$-smooth with respect to $\phi(s)$, Algorithm~\ref{algo} computes a $\lambda \left(1 + \epsilon\right)$-approximate equilibrium for every congestion game with non-decreasing cost functions, in number of steps which is polynomial in the number of players, $\Delta:=\frac{f(N)}{f(1)}$, $\lambda$, and $1/\epsilon$.
	\end{theorem}
\begin{proof}The algorithm partitions the players into blocks $B_1, B_2,\ldots,B_{\hat{z}}$ such that, a player 
\begin{align*}
	u \in B_i \Leftrightarrow \mathit{\ell_u} \in \left(b_{i+1},b_{i}\right],
\end{align*}
where \[b_i := \mathit{\ell_{max}}\left(2\Delta N^{2c+2}\right)^{-i+1} \text{and}~b_{i+1} := \mathit{\ell_{max}}\left(2\Delta N^{2c+2}\right)^{-i}\] define the boundaries of the block $B_i$. The partitioning is such that, the algorithm partitions the players to $\hat{z} = 1 + \lceil\log_{2\Delta N^{2c+2}}\left(\ell_{max}\slash \ell_{min}\right)\rceil \le N$ blocks where for any block $B_i$ the ratio $b_i\slash b_{i+1} = 2\Delta N^{2c+2}$ and $\Delta = \max_{e \in E} \frac{f_e(N)}{f_e(1)}$. Note that, for cost functions which are polynomials of maximum degree $d$ with non negative coefficients, $\Delta$ is polynomial in the number of players. Next, the algorithm enforces every player $u \in \mathcal{N}$ to play their optimistic strategy by which they incur a cost that could be at most $\Delta b_i$. This also defines the initial state of the game denoted by $s^0$, where $s^i$ denotes the state of the game after the phase $i$. The sequence of  moves in the game is divided into multiple phases determined by the player blocks. The phases of the game progresses from $1 \rightarrow \hat{z}-1$. During a phase $i$ of the game, only players in the block $B_i$ and $B_{i+1}$ make moves. Particularly, players in $B_i$ make their $p$-move using the original cost function $f$ and the players in $B_{i+1}$ make their $q$-move, but now using the modified cost function $f^\prime$ that satisfies Lemma~\ref{lem:subset} for some $\lambda>0$. 

For player $u \in \mathcal{N}$ a deviation to a strategy $s_u^\prime$ is called a $p\text{-move}$ if $c_u(s_u^\prime, s_{-u}) < \frac{c_u(s)}{p}$. Similarly a $q$-move with respect to the modified cost functions is defined as a move with $c'_u(s_u^\prime, s_{-u}) < \frac{c'_u(s)}{q}$. A phase $i$ is considered to be complete at a state $s^i$, if $\forall u \in B_i,~c_u(s^i) \le p\cdot c_u(s_u^\prime, s_{-u}^i)$ i.e., the players in $B_i$ are in a $p$-equilibrium. Similarly, $\forall u \in B_{i+1}, ~c'_u(s^i) \le q\cdot c'_u(s_u^\prime, s_{-u}^i)$ i.e., the players in $B_{i+1}$ are in a $q$-equilibrium w.r.t. the modified cost functions. All the other players i.e., $\mathcal{N}\setminus (B_i \cup B_{i+i})$ are frozen to their strategy associated with the phase $i-1$. Also, note that players in a block $B_i$ are frozen to their optimistic strategy $\mathcal{BR}_u(0)$ until phase $i-1$. The players involved during a phase $i$ are denoted by $R_i$. Now since during the phase $i$ only players in $R_i$ make their best-response moves, the latency introduced by these players on a particular resource of the game will be denoted as $f_e^{R_i}$. Moreover, since the players $\mathcal{N}\setminus R_i$ are frozen to their strategy of phase $i-1$, the latency incurred by a player $u \in R_i$ using a resource $e \in E$ can be expressed as $f_e^{R_i}(n_e^{R_i}(s)) = f_e(n_e^{R_i}(s) + n_e^{\mathcal{N}\setminus R_i}(s))$, where $n_e^{R_i}(s)$ denotes the number of players $u \in R_i$ in the state $s$ using the resource $e$ and $n_e^{\mathcal{N}\setminus R_i}(s)$ denotes the number of players on the resource $e$ in the state $s$ that do not participate in the phase $i$ of the game. Furthermore, the potential amongst the players in $R_i$ will be denoted as $\Phi_{R_i}$. 

  	 Here, we have to take into account that the game played by the players from $B_i \cup B_{i+1}$ in phase $i$ is no longer a potential game as the players use different cost functions. However, we can show that the strong smoothness condition of Lemma~\ref{lem:subset} guarantees that the  values of the modified cost functions $f^\prime$ can be conveniently bounded.

	\begin{lemma}
\label{lemma:fprime}
             Let $f^\prime$ to be the set of modified cost functions satisfying strong $(\lambda, 0)$-smoothness condition for some $\lambda>0$ and $f$ to be the original cost functions. Then for all $i\ge 1$, \[f_e(i) \le f_e'(i) \le \lambda  f_e(i).\]
\end{lemma}
\begin{proof}Using the strong smoothness condition of Lemma~\ref{lem:subset} and setting $n=0$, $m=1$, and $z=i-1$ gives, \[\lambda f_e(i) \ge f_e'(i). \] Furthermore, with $m=0$, $n=1$, and $z=i-1$ we have that, \[f_e'(i) \ge  f(i).\]
\end{proof}

	To bound the stretch of any (sub-) game in a $q$-approximate equilibrium  the following lemma is useful. In its proof we handle the modified cost functions which then leads to value of $\theta(q):= \frac{\lambda}{1 + N\lambda \frac{1-q}{q}}$ (cf. Algorithm~\ref{algo}) that depends on the stretch $\lambda$ of the modified cost functions, instead of the original ones. We remark that for this lemma, the  property that the induced subgames are also smooth (Lemma~\ref{lem:subset}) is crucial.
	
	\begin{lemma}\label{C_L_3}
		Let $s$ be any $q$-approximate equilibrium with respect to the modified cost function and $s^*$ be a strategy profile with minimal potential. Then for every $F \subseteq \mathcal{N}$, $\phi_F(s) \le \theta(q)\cdot \phi_F(s^*).$
	\end{lemma}
\begin{proof}Let $c'_u$ and $c_u$ denote the cost of a player $u \in F$ using the modified cost function $f'$ and the original cost function $f$, respectively. From the definition of $q$-approximate equilibrium we have that,\[c'_u(s) \le q\cdot c'_u(s_u^*, s_{-u}).\]
		Then, using Lemma \ref{lemma:fprime},  
		\begin{align*}
			c'_u(s_u^*, s_{-u}) - c'_u(s) 
			&\ge \frac{1-q}{q} c'_u(s) \ge   \frac{1-q}{q}\lambda c_u(s) 
			\ge  \frac{1-q}{q} \lambda\Phi_F(s),
		\end{align*}
		summing the above inequality for all players $u \in F$ gives, 
		\[\sum_{u\in F}\left(c'_u(s_u^*, s_{-u}) - c'_u(s)\right) \ge    \frac{1-q}{q}N\lambda \Phi_F(s). \tag{2}\]
		
		Then, by the smoothness condition of Lemma~\ref{lem:subset} and (2) we have,
		\begin{align*}
		    \Phi_F(s) 
		    &\le \lambda\cdot \Phi_F(s^*) - \Bigg(\sum_{u\in F}\left(c'_u(s_u^*, s_{-u}) - c'_u(s)\right)\Bigg)
			\\&\le \lambda\cdot \Phi_F(s^*) -  \frac{1-q}{q} N \lambda \Phi_F(s).
			\\
		    \left( 1 +    \frac{1-q}{q}N\lambda\right) \Phi_F(s) &\le  \lambda\cdot \Phi_F(s^*)\\
		   \Phi_F(s) &\le  \frac{\lambda}{1 +   \frac{1-q}{q} N\lambda }\cdot \Phi_F(s^*).\tag{3}
		\end{align*} 
		Setting $\theta(q) = \frac{\lambda}{1 + N\lambda \cdot \frac{1-q}{q}}$ in (3)
		concludes the proof.
	\end{proof}

		\begin{claim}[Caragiannis et al.~\cite{Caragiannis2011}]\label{F_C_1}
			For any state $s$ of a congestion game with a set of players $\mathcal{N}$, a set of resource $E$ and latency functions $(f_e)_{e\in E}$, it holds that \[\sum_{e\in E}f_e(n_e(s)) \le \phi(s) \le \sum_{u\in \mathcal{N}} c_u(s).\]
		\end{claim}
		\begin{lemma}[Caragiannis et al.~\cite{Caragiannis2011}]\label{F_L_1}
			Let $s$ be a state of the congestion game $\mathcal{G}$ with a set of players $\mathcal{N}$ and let $F\subseteq \mathcal{N}$. Then, $\phi(s) \le \phi_F(s) + \phi_{\mathcal{N}\setminus F}$ and $\phi(s) \ge \phi_F(s)$.
		\end{lemma}
		\begin{lemma}[Caragiannis et al.~\cite{Caragiannis2011}]\label{F_L_2}
			Let $c(u)$ denote the cost of player $u \in R_i$ just after making his last move within phase $i$. Then, \[ \phi_{R_i}(s^i) \le \sum_{u\in R_i}c(u).\]
		\end{lemma}
We now bound the potential of the set of players $R_i \subseteq B_i \cup B_{i+1}$ that move in phase $i$. Most importantly, the players of $B_i$, were in an $q$-approximate equilibrium with respect to $c'_u$ at the end of the previous round. Hence, for every subset of $B_i$, we can exploit Lemma~\ref{C_L_3} to obtain a small upper bound on the potential amongst players $R_i$ participating in a phase $i$ at the beginning of the phase. 
Recall that for a phase $i$, $b_i := \mathit{\ell_{max}}\left(2\Delta N^{2c+2}\right)^{-i+1}$ and $s^i$ denotes the state of the game after the execution of phase $i$. 
 		\begin{lemma}
 			\label{C_L_7}
 			For every phase $i\ge 2$, it holds that $\phi_{R_i}(s^{i-1}) \le \frac{b_i}{N^c}$.
 		\end{lemma}
 		\begin{proof}
 		Let us assume that inequality does not hold and $\phi_{R_i}(s^{i-1}) > \frac{b_i}{N^c}$. Then, we show that the players $u \in R_i \cap B_i$ were not in a $q$-approximate equilibrium w.r.t. the modified cost functions $f'$
 		in the phase $i-1$ and thus violating the dynamics of the algorithm.
			
			Note that the players in the block $B_{i+1}$ have not moved until the phase $i-1$ and are in their optimistic strategy $\mathcal{BR}_u(0)$ as per the initial settings of the algorithm. As per the definition of blocks the cost incurred by a player $u \in B_{i+1}$ is at most $\Delta b_{i+1}$. So the total cost incurred by all the players in $R_i \cap B_{i+1}$, i.e., $\sum_{u \in R_i \cap B_{i+1}}\Delta b_{i+1} \le N\Delta b_{i+1}$. The potential amongst players in $R_i \cap B_{i+1}$ is bounded by,
			\begin{align*}
				\Phi_{R_i\cap B_{i+1}}(s^{i-1})\le N\Delta b_{i+1}. \tag{4}
			\end{align*}
			Using Lemma~\ref{F_L_1}, (4), and our assumption on $\Phi_{R_i}(s^{i-1})$  we get,
			\begin{align*}
				\Phi_{R_i\cap B_i}(s^{i-1})
				&\ge \Phi_{R_i}(s^{i-1}) - \Phi_{R_i\cap B_{i+1}}(s^{i-1})
				\\&>\frac{b_i}{N^c}-N\Delta b_{i+1}
				\\&=\left(\frac{2\Delta N^{2c+2}}{N^c}-N\Delta\right)b_{i+1}
				\\&\ge N^{c+1}\Delta b_{i+1}.	\tag{5}
			\end{align*}

			Let us denote by $c(u)$ the latency of a player $u \in R_i \cap B_i$ after he made his last move during the phase $i$.
			So, the change in potential contributed by the player $u$ in the phase $i$ is at least $(p-1)c(u)$. Let us denote by $\xi_i$ the decrease of potential due to the moves of the players in $B_{i+1}$ in phase $i$. Note that $\xi_i$ could be negative as the players of $B_{i+1}$ use the modified cost functions. The change in potential due to the moves by all the players in $R_i$ 
			is given by $(p-1) \sum_{u\in R_i\cap B_i}c(u) + \xi_i$ and we can bound
			\begin{align*}
				&(p-1)\sum_{u\in R_i\cap B_i}c(u) 
				\\&\le \Phi_{R_i}(s^{i-1}) - \Phi_{R_i}(s^i) - \xi_i	
				\\&\le \Phi_{R_i\cap B_i}(s^{i-1}) + \Phi_{R_i\cap B_{i+1}}(s^{i-1}) - \Phi_{R_i}(s^i)  - \xi_i
				\\&\le \Phi_{R_i\cap B_i}(s^{i-1}) + N\Delta b_{i+1} - \Phi_{R_i}(s^i) - \xi_i
				\\&<\left(1+\frac{1}{N^c}\right)\Phi_{R_i \cap B_i}(s^{i-1}) - \Phi_{R_i}(s^i)  - \xi_i. \tag{6}
			\end{align*}
			To account in the change of potential $\xi_i$ from $s^{i-1}$ to $s^i$ due the players in $B_{i+1}$, we observe that the latency of a player $u \in R_i \cap B_{i+1}$ was at most $c_u(s^{i-1}) \le  \Delta b_{i+1}$ as he was put by the algorithm on $\mathcal{BR}(0)$. By Lemma \ref{lemma:fprime}, his cost with respect to the modified cost function on this strategy are $c'_u(s^{i-1}) \le  \lambda \Delta b_{i+1}$. Since, he may always switch back to this strategy, his cost in $s_i$ can be bounded by 
			
			\[c_u(s^{i}) \le c'_u(s^{i}) < q  \lambda \Delta b_{i+1}.\]
			
			This yields a bound on the change  of the potential of \[  \xi_i > - q N \lambda \Delta b_{i+1}.\]

			Now we can bound the potential in $s^i$ by the latency  of the players. We then can use  inequality (6) for the players in $R_i \cap B_{i}$. By Lemma \ref{lemma:fprime}, the latency of a player $u \in R_i \cap B_{i+1}$ after he made his last $q$-move during the phase $i$ is at most $\lambda\Delta b_{i+1}$.
			\begin{align*}
				\Phi_{R_i}(s^i) 
				&\le \sum_{u \in R_i}c(u)
				\\&= \sum_{u \in R_i \cap B_{i+1}}c(u) + \sum_{u \in R_i\cap B_i}c(u)
				\\&< N \lambda \Delta b_{i+1} + \frac{1}{p-1}\left(1+\frac{1}{N^c}\right)\Phi_{R_i \cap B_i}(s^{i-1})	
				\\&\quad-\frac{1}{p-1}\Phi_{R_i}(s^i)
				-\frac{1}{p-1} \xi_i
				\\&\le \frac{\lambda }{N^c} \Phi_{R_i \cap B_i}(s^{i-1}) + \frac{1}{p-1}\left(1+\frac{1}{N^c}\right)\Phi_{R_i \cap B_i}(s^{i-1})
				\\&\quad-\frac{1}{p-1}\Phi_{R_i}(s^i)
				+ \frac{q}{p-1}  N \lambda \Delta b_{i+1}
			\\&\le \frac{\lambda }{N^c} \Phi_{R_i \cap B_i}(s^{i-1}) + \frac{1}{p-1}\left(1+\frac{1}{N^c}\right)\Phi_{R_i \cap B_i}(s^{i-1})
				\\&\quad-\frac{1}{p-1}\Phi_{R_i}(s^i)
				+ \frac{q}{p-1} \frac{\lambda }{N^c} \Phi_{R_i \cap B_i}(s^{i-1}) 
			\\&\le \frac{\lambda(p-1) }{(p-1)N^c} \Phi_{R_i \cap B_i}(s^{i-1}) + \frac{1}{p-1}\left(1+\frac{1}{N^c}\right)\Phi_{R_i \cap B_i}(s^{i-1})
				\\&\quad-\frac{1}{p-1}\Phi_{R_i}(s^i)
				+ \frac{q}{p-1} \frac{\lambda }{N^c} \Phi_{R_i \cap B_i}(s^{i-1}) 
			\\&\le \frac{1}{p-1}\left(1+\frac{(p-1)\lambda + 1 +q\lambda}{N^c}\right)\Phi_{R_i\cap B_i}(s^{i-1})
				\\&\quad-\frac{1}{p-1}\Phi_{R_i}(s^i)
		\end{align*}
		equivalent to,
		\begin{align*}
		    \frac{p}{p-1}\Phi_{R_i}(s^i)&< \frac{1}{p-1}\left(1+\frac{(p-1)\lambda + 1 +q\lambda}{N^c}\right)\Phi_{R_i\cap B_i}(s^{i-1})
		    \\&=\frac{p}{p-1}\left(\frac{1}{p}+\frac{(p-1)\lambda + 1 +q\lambda}{pN^c}\right)\Phi_{R_i\cap B_i}(s^{i-1})
		    \\&<\frac{p}{p-1}\left(\frac{1}{p}+\frac{\lambda}{N^c}+\frac{1 +q\lambda}{pN^c}\right)\Phi_{R_i\cap B_i}(s^{i-1})
		    \\&<\frac{p}{p-1}\left(\frac{1}{p}+\frac{\lambda}{N^c}+\frac{1 +q}{N^c}\right)\Phi_{R_i\cap B_i}(s^{i-1}).
		\end{align*}
		Therefore,
		\begin{align*}
		    \Phi_{R_i}(s^i)&<\left(\frac{1}{p}+\frac{1 +q+\lambda}{N^c}\right)\Phi_{R_i\cap B_i}(s^{i-1}).\tag{7}
		\end{align*}

			Observe that during the phase $i-1$, the players in the block $R_i \cap B_{i+1}$ have not deviated from their initial strategy of $\mathcal{BR}_u(0)$ in the state $s^{i-1}$. However, this cannot be guaranteed in the state $s^i$ where the players in $R_i \cap B_{i+1}$ could have made their best-response $q$-moves. Now in order to compare the potential amongst the players in $R_i \cap B_i$ in the state $s^{i-1}$ and $s^i$, it is important that the players in $R_i \cap B_{i+1}$ have the same strategy as it had in phase $i-1$. So, we construct the following thought experiment. Let $\hat{s}$ be a state where player in $R_i \cap B_i$ play their strategy in $s^i$ and players $u \in \mathcal{N}\setminus (R_i \cap B_i)$ play their strategy in $s^{i-1}$. Since, the cost incurred by players in $R_i \cap B_{i+1}$ in the state $s^i$ after deviating to their strategy in $s^{i-1}$ is at most $N\lambda\Delta b_{i+1}$.
 
			The potential amongst the players in $R_i$ in the state $\hat{s}$ is given by,
			\begin{align*}
				\Phi_{R_i}(\hat{s})
				&\le \Phi_{R_i\cap B_i}(s^i) + N\lambda\Delta b_{i+1}
				\\&\le \Phi_{R_i}(s^i) + N\lambda\Delta b_{i+1}. \tag{8}
			\end{align*}

			Using Lemma~\ref{F_L_1} we get,
			\begin{align*}
				\Phi_{R_i\cap B_i}(\hat{s})
				&\le \Phi_{R_i}(\hat{s})
				\\\text{Applying inequality (8) we get,}
				\\&\le \Phi_{R_i}(s^i) + N\lambda\Delta b_{i+1}
		        \\\text{Then from inequality (5) and (7),}
		    	\\&<\left(\frac{1}{p}+\frac{1+q+2\lambda}{N^c}\right)\Phi_{R_i\cap B_i}(s^{i-1})
				\\&=\frac{1}{\theta(q)}\Phi_{R_i\cap B_i}(s^{i-1}).
			\end{align*}
			The last equality follows from the definition of $p$ in Algorithm~\ref{algo}.

			If $s^*$ were to be the state in which the game attained its global minimum, then the last inequality effectively means that the potential amongst the players in $R_i\cap B_i$  in state $s^*$ i.e., $\Phi_{R_i}(s^*)$ is strictly smaller than $\frac{1}{\theta(q)}\Phi_{R_i\cap B_i}(s^{i-1})$ which violates the claim in Lemma~\ref{C_L_3} to conclude that players in $R_i\cap B_i$ are not in a $q$-equilibrium at the end of the phase $i-1$. Hence, contradicting our assumption.
		\end{proof}

To analyze convergence, we have to take into account the fact that players use different latency functions and, hence, convergence is no longer guaranteed by Rosenthal's potential function.  However, it turns out that the Rosenthal potential with respect to the modified cost functions can serve as an approximate potential function, i.e., it also decreases for the $p$-moves of players using the original cost functions. 
		\begin{lemma}\label{approx_potential}
		The Rosenthal potential $\widetilde{\phi}$ with respect to the modified cost functions $f^\prime$ is a p-approximate potential function with respect to the original cost function $f$. That is,
		\[c_u(s'_u, s_{-u}) < \frac{1}{p} c_u(s) ~\text{implies}~ \widetilde{\phi}(s'_u,s_{-u}) < \widetilde{\phi}(s),\] where, \[ \widetilde{\phi}(s) := \sum_{e \in E} \widetilde{\phi}_e(n_e(s))= \sum_{e \in E}\sum_{i =1}^{n_e(s)}f_e^\prime(i).\]
		\end{lemma}
		\begin{proof}To simplify notation let $n_e := n_e(s)$. From Lemma~\ref{lemma:fprime} we know that,
\[\widetilde{\phi_e}(n_e +1) - \widetilde{\phi_e}(n_e) = f^\prime(n_e+1) \le \lambda f(n_e+1),\]
and
\[f^\prime(n_e+1) \ge  f(n_e+1).\]
Therefore,
\[ 1 \le \frac{\widetilde{\phi}_e(n_e +1) - \widetilde{\phi}_e(n_e)}{f_e(n_e +1)} \le \lambda. \]

Using the above inequalities, the change in potential function $\widetilde{\phi}$ due to player $u$ making a $p$-move with respect to $f$, i.e.,
\begin{align*}
    \widetilde{\phi}_e(s'_u,s_{-u}) - \widetilde{\phi}_e(s) &= \sum_{e \in E}  \widetilde{\phi}_e(s'_u,s_{-u}) - \widetilde{\phi}_e(s)
    \\&=\sum_{e \in s'_u\setminus s_u}  \widetilde{\phi}_e(n_e +1) - \widetilde{\phi}_e(n_e) + \sum_{e \in s_u\setminus s'_u}  \widetilde{\phi}_e(n_e - 1) - \widetilde{\phi}_e(n_e)
    \\&\le\sum_{e \in s'_u\setminus s_u} \lambda \cdot f_e(n_e+1) - \sum_{e \in s_u\setminus s'_u}   f_e(n_e)
    \\&\le\lambda \left(\sum_{e \in s'_u\setminus s_u} f_e(n_e  +1)+ \sum_{e \in s'_u \cap s_u} f_e(n_e)\right) 
    \\&\qquad- \left(\sum_{e \in s_u\setminus s'_u}   f_e(n_e)+ \sum_{e \in s'_u \cap s_u} f_e(n_e)\right)
    \\&=\lambda \cdot c_u(s'_u, s_{-u}) - c_u(s)
    \\&\le p \cdot c_u(s'_u, s_{-u}) - c_u(s).
\end{align*}
The last inequality is due to the choice of $p$ in Algorithm~\ref{algo} such that it is slightly larger than $\lambda$.
		\end{proof}
		
		The following lemma exhibits a even stronger property. It shows that $p$-moves with respect to the original cost functions are $q$-moves with respect to the modified cost functions.
		\begin{lemma}\label{L_patleastq}
		Let $u \in \mathcal{N}$ be a player that makes a p-move with respect to the original cost function $f$. Then,
		\[p\cdot c_u(s^\prime_u, s_{-u}) - c_u(s)\ge q \cdot c^\prime_u(s^\prime_u, s_{-u}) - c^\prime_u(s),\] where $c_u$, and $c^\prime_u$ are the cost of the player $u$ with respect to $f$ and $f^\prime$, respectively.
		\end{lemma}
		\begin{proof}Let us recall the definition of $p$, $q$, and $\theta(q)$ in Algorithm~\ref{algo}, \[ p := \left( \frac{1}{\theta(q)} - \frac{1+2\lambda+q}{N^c}\right)^{-1}\]
		\[q := \left( 1 + \frac{1}{N^c}\right)\]
		\[\theta(q) := \frac{\lambda}{1 + \frac{1-q}{q}N\lambda}\]
		Observe that $p \ge \theta(q)$. Therefore,
		         \begin{align*}
		             p\cdot c_u(s^\prime_u, s_{-u}) - c_u(s) &\ge \theta(q) \cdot c_u(s^\prime_u, s_{-u}) - c_u(s)
		             \\&= \frac{\lambda}{1 + \frac{1-q}{q}N\lambda} \cdot c_u(s^\prime_u, s_{-u}) - c_u(s)
		             \\&=  \frac{q\lambda}{ 1 + \frac{1}{N^c}(1 - N\lambda)} \cdot c_u(s^\prime_u, s_{-u}) - c_u(s)
		             \\&\ge q\lambda \cdot c_u(s^\prime_u, s_{-u}) - c_u(s).
		             \\\text{Then, from Lemma~\ref{lemma:fprime} we have that,}
		             \\p\cdot c_u(s^\prime_u, s_{-u}) - c_u(s)&\ge q \cdot c^\prime_u(s^\prime_u, s_{-u}) - c^\prime_u(s).
		         \end{align*}
		\end{proof}
		
		  Using Lemma~\ref{C_L_7} and~\ref{L_patleastq}, we can  bound the runtime which depends on $\Delta$ to allow for arbitrary non-decreasing functions. 
	    \begin{lemma}
			\label{C_L_8}
			The algorithm terminates after at most $\mathcal{O}(\lambda\Delta^3N^{5c+5})$ best-response moves.
		\end{lemma}
		  \begin{proof}The proof follows from Lemma~\ref{approx_potential} and Lemma~\ref{L_patleastq}. Again, we denote $f^\prime$ to be the modified cost functions 
		  and $f$ to be the original cost functions.

			Let us recall that the algorithm partitions the sequence of best-response moves in the game into $\hat{z} - 1$ phases, where $\hat{z} = 1 + \lceil\log_{2\Delta n^{2c+2}}\left(\mathit{\ell_{max}}\slash\mathit{\ell_{min}}\right)\rceil\le N$. In a phase $i \in \{1,\dots, \hat{z} - 1\}$, players in block $R_i \cap B_i$ make their $p$-move with respect to $f$ and players in block $R_i \cap B_{i+1}$ make their $q$-move with respect to $f'$. We will bound the number of $p$-moves and $q$-moves in any given phase $i$, using the potential function with respect to the modified cost function $f^\prime$. Lemma~\ref{approx_potential} shows that when players in the block $B_{i}$ make their $p$-moves, they also reduce the potential function with respect to the modified cost functions $f^\prime$. Lemma~\ref{L_patleastq} shows that the change in cost of a player due to a $p$-move with respect to the function $f$ is at least the change in cost due to a $q$-move with respect to $f^\prime$. Therefore, in order to bound total number of moves in any given phase $i$, it is sufficient to assume that players in block $B_i$ make $q$-moves with respect to $f^\prime$ instead of $p$-moves. 
			
			Using these we now bound the maximum number of best response moves in a phase,\\
			\textit{Phase $i=1$}:\\
			Let us assume that all players in the phase $R_1$ have a $q$-move.

			Define, 
			\[\Delta^\prime =\max_{e \in E, n \in N}\frac{f_e^\prime(n)}{f_e(1)} \le \max_{e \in E, n \in N}\frac{\lambda\cdot f_e(n)}{f_e(1)},\]

			where the inequality follows from Lemma~\ref{lemma:fprime}. Then, for any player $u \in R_1$, the maximum latency incurred by the player at the beginning of phase with respect to $f^\prime$ is at most $\Delta^\prime\cdot b_1$. Observe that the maximum potential associated with the sub-game in the phase $R_1$ with respect to the modified cost function $f^\prime$ is then at most $N\Delta^\prime b_1$.

			Also, the minimum latency experienced by the players in $R_1$ is at least $b_3$. So, when a player in $u \in R_1$ makes a best-response move, he must reduce the potential by at least $(q-1)b_3$. Then, using the fact that $b_i = 2\Delta N^{2c+2}b_{i+1}$, we obtain the number of best response moves amongst the players in $R_1$ to be at most,
			\begin{align*}
				&\frac{N\Delta^\prime b_1}{(q-1)b_3}=\frac{N\Delta^\prime\left(4\Delta^2 N^{5c+4}\right)}{N^{c}(q-1)} \le4 \lambda\Delta^3N^{5c+5}. \tag{9}
			\end{align*}
			\textit{Phase $i\ge2$}:
			Again, let us assume that all players in the $R_i$ have a $q$-move. Lemma~\ref{C_L_7} shows that for each phase $i\ge2$, the potential amongst the players $R_i$ participating in the phase $i$ at the beginning of the phase i.e., $\Phi_{R_i}(s^{i-1})$ is at most $\frac{b_i}{N^c}$.

			Therefore, due to Lemma~\ref{lemma:fprime} the potential with respect to the modified cost function is then at most $\frac{\lambda\cdot b_i}{N^c}$. By the definition of blocks the minimum latency that a player would incur is at least $b_{i+2}$. So, when a player $u$ makes his best-response move during phase $i$, he would reduce the potential of the sub-game and thus the players in $R_i$ by at least $(q-1)b_{i+2}$. Hence, using the fact that $b_i = 2\Delta N^{2c+2}b_{i+1}$, we obtain the number of best response moves amongst the players in $R_i$ to be at most,
			\begin{align*}
				&\frac{\lambda\cdot b_i}{N^c(q-1)b_{i+2}}= \frac{\lambda \left(4\Delta^2 N^{4c+4}\right)}{N^c(q-1)}\le
				4\lambda\Delta^2 N^{4c+4}. \tag{10}
			\end{align*}
			From (9) and (10) we get the desired upper bound on the number of best-response moves in the game to be at most $\mathcal{O}(\lambda\Delta^3N^{5c+5})$.
		\end{proof}
		
		The next lemma shows that when players involved in phases $i \ge 2$ make their  moves, they do not increase the cost of players in the blocks $B_1, B_2,\cdots,B_{i-1}$ significantly. 
\begin{lemma}
			\label{C_L_9}
			Let $u$ be a player that takes part in the phase $t \le  i$, 
			then it holds that, \[ c_u(s^{i+1}) \le c_u(s^i) + \frac{b_{i+1}}{N^c} + N\lambda \Delta  b_{i+2}.\]
		
		\end{lemma}	
		\begin{proof}We derive the proof by showing that if the increase in cost is greater than $\frac{b_{i+1}}{N^c} + N\lambda \Delta  b_{i+2}$, then it violates the fact that $\Phi_{R_{i+1}}(s^i) \le \frac{b_{i+1}}{N^c}$.
			
			Now, let us assume that $\exists u \in B_i$ for whom the claim does not hold i.e., 
			\begin{align*}
				c_u(s^{i+1}) > c_u(s^i) + \frac{b_{i+1}}{N^c} +N\lambda \Delta  b_{i+2}. \tag{11}
			\end{align*}
			This implies that 
			there exists a set of resources $C\subseteq s_u$ such that for each $e \in C$ it is used by at least one player in $R_{i+1}$ in the state $s^{i+1}$, thus contributing to the increase in cost of the player $u$. Then from (11) we have, 
			\begin{align*}
				\sum_{e \in C}f_e(n_e(s^{i+1})) > \frac{b_{i+1}}{N^c} +N\lambda \Delta  b_{i+2}.
			\end{align*}
			Then,
			\begin{align*}
				\Phi_{R_{i+1}}(s^{i+1}) > \frac{b_{i+1}}{N^c} +N\lambda \Delta  b_{i+2}.
			\end{align*}
			
			As the players in $R_{i+1} \cap B_{i+2}$ might have increased (or decreased) $\Phi_{R_{i+1}}$ by at most $N\lambda \Delta  b_{i+2}$ and the players  $R_{i+1} \setminus B_{i+2}$ only decreased the potential, we know that
			\begin{align*}
				\Phi_{R_{i+1}}(s^{i}) 
				&\ge \Phi_{R_{i+1}}(s^{i+1}) - N\lambda \Delta  b_{i+2}
				\\&> \frac{b_{i+1}}{N^c} + N\lambda \Delta  b_{i+2} - N\lambda \Delta  b_{i+2}
				\\&
				=\frac{b_{i+1}}{N^c}
			\end{align*}
			The last inequality violates Lemma \ref{C_L_7}. Hence, this contradicts our assumption and thus the claim holds for the player $u$.
		\end{proof}
		\begin{lemma}[Caragiannis et al.~\cite{Caragiannis2011}]
			\label{C_L_10}
			Let $u$ be a player that takes part in the phase $t \le i$ of the congestion game $\mathcal{G}$ and let $s_u^\prime$ be any strategy other than the one assigned by the algorithm during the phase $t$ of the game, then it holds that, \[c_u(s_u^\prime, s_{-u}^i) \le c_u(s_u^\prime, s_{-u}^{i+1}) + \frac{b_{i+1}}{N^c}.\]
		\end{lemma}
		\begin{proof}Assume the claim does not hold for some player $u$ i.e., 
			\begin{align*}
				c_u(s_u^\prime, s_{-u}^i) > c_u(s_u^\prime, s_{-u}^{i+1}) + \frac{b_{i+1}}{N^c} .
			\end{align*} 
		  This means that during the phase $i$, there exists a subset of resources $C\subseteq s_u^\prime$ such that for each $e \in C,~\exists u^\prime \in R_{i+1}$ who used the resource $e$ in the state $s^i$ but not in $s^{i+1}$ and thus contributed to cost incurred by the player $u$ during the phase $i$ in the state $s^i$. Giving, \[\sum_{e \in C}f_e(n_e(s_u^\prime,s_{-u}^i)) > \frac{b_{i+1}}{N^c}.\] Furthermore, the cost of these resources yield a lower bound on the potential at the beginning of the phase:
			\begin{align*}
				\Phi_{R_{i+1}}(s^i) 
				&\ge \sum_{e \in C}f_e(n_e(s_u^\prime,s_{-u}^i))
				\\& >\frac{b_{i+1}}{N^c}. 
			\end{align*}
			The last inequality violates Lemma \ref{C_L_7}. Hence, this contradicts our assumption and thus the claim holds. 
		\end{proof}
		\begin{lemma}
			\label{C_L_11}
			Let u be a player in the block $B_t$, where $t \le \hat{z}-2$. Let $s_u^\prime$ be a strategy different from the one assigned to u by the algorithm at the end of the phase t. Then, for each phase $i\ge t$, it holds that, $c_u(s^i) \le p\cdot c_u(s_u^\prime, s_{-u}^i) + \frac{2p+1}{N^c}\sum_{k=t+1}^{i}b_k.$
		\end{lemma}
				\begin{proof}For the proof we use Lemma \ref{C_L_9} recursively to obtain the first inequality. The second inequality follows from the fact that there was no improving $p$-move to $s'_u$ for the player phase $t$ to $s'_u$. The third inequality follows from Lemma \ref{C_L_10}. The fourth inequality from the definition of $b_i$.
			\begin{align*}
				c_u(s^{i}) 
				&\le c_u(s^t) + \sum_{k=t+1}^{i} \left(\frac{b_k}{N^c} + {N\lambda \Delta b_{k+1}}\right)
				\\&\le p\cdot c_u(s_u^\prime, s_{-u}^{t}) + \sum_{k=t+1}^{i} \left(\frac{b_k}{N^c} + {N\lambda \Delta b_{k+1}}\right)
				\\&\le p\left(c_u(s_u^\prime, s_{-u}^{i}) + \sum_{k=t+1}^{i}\frac{b_{k}}{N^c}\right) + \sum_{k=t+1}^{i} \left(\frac{b_k}{N^c} + {N\lambda \Delta b_{k+1}}\right)	
				\\&\le p\left(c_u(s_u^\prime, s_{-u}^{i}) + \sum_{k=t+1}^{i}\frac{b_{k}}{N^c}\right) + \sum_{k=t+1}^{i} \left(\frac{b_k}{N^c} + \frac{\lambda }{2 N^{2c+1}} b_{k}\right)
				\\&\le p\cdot c_u(s_u^\prime, s_{-u}^{i}) + (2p+1)\sum_{k=t+1}^{i}\frac{b_{k}}{N^c} 
				.
			\end{align*}
		\end{proof}

        As no players' costs and alternatives is significantly influenced by moves in later blocks, they remain in an approximate equilibrium which can be used to finally prove the correctness of the algorithm.
		\begin{lemma}
			\label{C_L_12}
			The state computed by the algorithm is a $p \left(1+\frac{5}{N^c}\right)$-approximate equilibrium.
		\end{lemma}
				\begin{proof}The idea behind the lemma is to show that after a player $u \in B_i$ has made his final best-response move during a phase $i$, he would be in a $p\left(1+\frac{5}{N^c}\right)$-approximate equilibrium at the end of the game i.e., cost incurred by him after the final phase of the game is, $c_u(s^{\hat{z} -1})\le $ $p\left(1+\frac{5}{N^c}\right)c_u(s_u^\prime, s_{-u}^{\hat{z} -1})$, where $s_u^\prime$ is any strategy.
			Now, for players participating in the last phase $\hat{z}-1$ of the game i.e., $u \in (B_{\hat{z} -1} \cup B_{\hat{z}})$, observe that at the end of the phase $\hat{z} -1$, players in block $B_{\hat{z} -1}$ are in a $p$-approximate equilibrium and players in the block $B_{\hat{z}}$ are in a $q$-approximate equilibrium with respect to the modified cost function $f^\prime$. Furthermore, due to Lemma~\ref{L_patleastq} players in the block $B_{\hat{z}}$ are also in a $p$-approximate equilibrium with respect to the original cost function $f$. Therefore, the lemma holds for players in  $u \in (B_{\hat{z} -1} \cup B_{\hat{z}})$.
			
			It is now left to show that for the players $u \in B_t$ where $1\le t \le \hat{z} -2$, after the final phase of the game, \[c_u(s^{\hat{z} -1}) \le p\left(1+\frac{5}{N^c}\right)c_u(s_u^\prime, s_{-u}^{\hat{z} -1}).\]
			
			By the definition of assignment of the players to blocks and the observation that the cost of a player cannot be less than the $\mathit{\ell_u}$, it holds that for any player $u \in B_t$ after the final phase of the game,
			\begin{align*}
				c_u(s_u^\prime, s_{-u}^{\hat{z} -1}) \ge b_{t+1}. \tag{12}
			\end{align*}
			By the definition of $b_i$, we have,
			\begin{align*}
				\sum_{k=t+1}^{\hat{z}}b_k \le 2b_{t+1}. \tag{13}
			\end{align*}
			Using inequalities (12), (13), and Lemma \ref{C_L_11} we get for $p \ge 1$,
			\begin{align*}
				c_u(s^{\hat{z} -1})
				&\le p\cdot c_u(s_u^\prime, s_{-u}^{\hat{z} -1}) + (2p+1)\sum_{k=t+1}^{\hat{z} -1}\frac{b_{k}}{N^c}
				\\&\le p\cdot c_u(s_u^\prime, s_{-u}^{\hat{z} -1}) + \frac{2(2p+1)}{N^c}c_u(s_u^\prime, s_{-u}^{\hat{z} -1})
				\\&\le p\left(1+\frac{5}{N^c}\right)c_u(s_u^\prime, s_{-u}^{\hat{z} -1}).
			\end{align*}
		\end{proof}
		
		Lemmas \ref{C_L_8} and \ref{C_L_12} conclude the proof of the Theorem~\ref{C_T_1} to show that for \[q = \left(1 + \frac{1}{N^c}\right)\] and \[p = \left(\frac{1}{\theta\left(q\right)} - \frac{1+2\lambda+q}{N^c}\right)^{-1}\]
		the algorithm computes a $\alpha$-approximate pure Nash equilibrium in polynomial time, where

        \begin{align*}
        \alpha
			&\le \left(\frac{1}{\theta(q)}-\frac{1+2\lambda+q}{N^c}\right)^{-1}\left(1+\frac{5}{N^c}\right)
			\\&=\frac{1}{\left( \frac{1 - \frac{N\lambda}{N^c+1}}{\lambda} - \frac{1+2\lambda + q}{N^c}\right)}\left(1+\frac{5}{N^c}\right)
			\\&=\frac{\lambda}{\left( 1 - \frac{N\lambda}{N^c+1} - \frac{\lambda(1+2\lambda + q)}{N^c}\right)}\left(1+\frac{5}{N^c}\right)
			\\&\le\frac{\lambda}{\left( 1 - \frac{N\lambda}{N^c+1} - \frac{\lambda(2\lambda+3)}{N^c}\right)}\left(1+\frac{5}{N^c}\right),
        \end{align*}
	choosing $c = 10 \log\left(\frac{\lambda}{\epsilon}\right)$ we can easily bound
	 \begin{align*}
         \frac{\lambda}{\left( 1 - \frac{N\lambda}{N^c+1} - \frac{\lambda(2\lambda+3)}{N^c}\right)}\left(1+\frac{5}{N^c}\right)
			&\le\frac{\lambda}{\left( 1 - \frac{\epsilon}{5} - \frac{\epsilon}{5}\right)}\left(1+\frac{\epsilon}{5}\right) \\ 
		&\le\ \lambda (1 + \epsilon).
		\end{align*}
		 This concludes the proof of correctness for Algorithm~\ref{algo}.
		
\end{proof}

\subsection{Improving the Approximation Factor}
In Section~\ref{analysis_sec} we proved that given a set of cost functions $f^\prime$ satisfying the strong smoothness condition of Lemma~\ref{lem:subset} for some $\lambda>0$, Algorithm~\ref{algo} computes a $\lambda(1+\epsilon)$-approximate pure Nash equilibrium. We now study how one could optimally choose the modified cost functions $f^\prime$ such that value of $\lambda$ is minimized.
 Observe that from Lemma~\ref{lem:subset}, for any resource $e \in E$, the modified cost functions $f_e'$ can be computed using a linear program.
	That is, for a objective function $\phi_e$ and a bound on the number of players $N$, finding functions $f'_e$ that minimize $\lambda$, can be easily solved by the following linear program LP$_\phi$ with the variables $f'_e(1),\ldots,f'_e(N)$, and $\lambda_e$.
    \begin{align*} 
		\min \lambda_e 
		\\\lambda_e \cdot \sum_{i=z+1}^{m+z} f_e(i) - m f'_e(n+z+1) + n f'_e(n+z) &\ge \sum_{i=z+1}^{n+z} f_e(i) &\forall (n+z), m \in [0,N]
		\\f'_e(n) &\ge 0   &\forall n \in [0,N+1]
	\end{align*}

    Choosing the output of the linear program LP$_\phi$ as the modified cost functions $f^\prime$ and setting value of $\lambda:=\max_{e \in E}\lambda_e$ in Algorithm~\ref{algo}, results in $\lambda(1+\epsilon)$-approximate pure Nash equilibrium.
    
	Observe that LP$_\phi$ is compact, i.e., the number of constraints and variables are polynomially bounded in the number of players. Hence, we state the following theorem. 
	\begin{theorem}
		Optimal resource cost functions $f'_e$ for objective functions $\phi_e$ can be computed in polynomial time.
	\end{theorem}

\noindent \textbf{Linear and Polynomial Cost Functions}

Observe that Theorem~\ref{C_T_1} holds for all congestion games with arbitrary non-decreasing cost functions. We now turn to the important class of polynomial cost functions with non-negative coefficients. 
We can use a standard trick to simplify the analysis for resources with polynomial cost functions of the form $f_e(x) = \sum_{i=0}^d a_i x^d$ by replacing such resources by $d+1$ resources with cost functions $a_0, a_1 x,\ldots, a_d x^d$ and adjust the strategy sets accordingly. Hence, by an additional scaling argument it suffices to only consider cost functions of the form $f_e(x)=x^i$ and $i \in \{1,\dots,d\}$.

Furthermore, we can show that for polynomials of small degree, it is sufficient to restrict the attention to the first $K=150$  values of the cost functions. 
Hence, we only need to solve a linear program of constant size. The following lemma states that for the larger values of $n$, appropriate values of $\lambda_d$, and $\nu$, we can easily obtain $(\lambda_d, 0)$-smoothness by choosing $f'(n)=\nu n^d$. 

\begin{lemma}
\label{lem:largeNPotential}
For $d \le 5$ and $n\ge  150$, the function $f'(n)= \nu n^d$ with $\nu = \sqrt[d+1]{\lambda_d}$ is strong $(\lambda_d, 0)$-smooth with respect to the potential function $\phi$ for an appropriate $\lambda_d$. 
\end{lemma}
   \begin{proof}For each degree $d$ and its associated value of $\lambda_d = \rho_d$ (Table~\ref{table:apxPNE}), we need to show that there exist a $\nu$ such that for the smoothness condition in Lemma~\ref{lem:subset}, the following holds for all $x\ge 150$,
\[\lambda_d \sum_{i=z+1}^{m+z}i^d-m \cdot \nu \cdot (x+1)^d +(x-z) \cdot \nu \cdot x^d - \sum_{i=z+1}^{x}i^d \ge 0. \]

For all $\lambda_d > 0$,
\begin{align*}
    &\lambda_d \sum_{i=z+1}^{m+z}i^d-m \cdot \nu \cdot (x+1)^d +(x-z) \cdot \nu \cdot x^d - \sum_{i=z+1}^{x}i^d
    \\&\ge\lambda_d \sum_{i=z+1}^{m+z}i^d-m \cdot \nu \cdot (x+1)^d +(x-z) \cdot \nu \cdot x^d - \int_{z+1}^{x+1}t^ddt
    \\&\ge \lambda_d \sum_{i=z+1}^{m+z}i^d-m \cdot \nu \cdot (x+1)^d +(x-z) \cdot \nu \cdot x^d - \frac{(x+1)^{d+1}}{d+1} + \frac{(z+1)^{d+1}}{d+1}
    \\&= \lambda_d \sum_{i=z+1}^{m+z}i^d-m \cdot \nu \cdot (x+1)^d + \nu \cdot x^{d+1} -  \nu \cdot z\cdot x^d - \frac{(x+1)^{d+1}}{d+1} + \frac{(z+1)^{d+1}}{d+1}
    \\&\ge \lambda_d \int_{z+1}^{m+z}t^ddt-m \cdot \nu \cdot (x+1)^d + \nu \cdot x^{d+1} -  \nu \cdot z\cdot x^d - \frac{(x+1)^{d+1}}{d+1} + \frac{(z+1)^{d+1}}{d+1}
    \\&= \lambda_d \frac{(m+z)^{d+1}}{d+1} - \lambda_d\frac{(z+1)^{d+1}}{d+1} -m \cdot \nu \cdot (x+1)^d + \nu \cdot x^{d+1} -  \nu \cdot z\cdot x^d - \frac{(x+1)^{d+1}}{d+1} 
    \\&\qquad+ \frac{(z+1)^{d+1}}{d+1}
\end{align*}
The above expression is minimized for,
\begin{align*}
    m = \sqrt[d]{\frac{\nu \cdot (x+1)^d}{\lambda_d}} - z.
\end{align*}
Substituting for $m$ we obtain,
\begin{align*}
    &\ge \lambda_d \frac{\left(\sqrt[d]{\frac{\nu \cdot (x+1)^d}{\lambda_d}}\right)^{d+1}}{d+1} - (\lambda_d-1)\frac{(z+1)^{d+1}}{d+1} - \left(\sqrt[d]{\frac{\nu \cdot (x+1)^d}{\lambda_d}} - z\right) \cdot \nu \cdot (x+1)^d   
    \\&\qquad + \nu \cdot x^{d+1} -\nu \cdot z\cdot x^d - \frac{(x+1)^{d+1}}{d+1} 
    \\&= \lambda_d \frac{\left(\sqrt[d]{\frac{\nu \cdot (x+1)^d}{\lambda_d}}\right)^{d+1}}{d+1} - (\lambda_d-1)\frac{(z+1)^{d+1}}{d+1} - \sqrt[d]{\frac{\nu \cdot (x+1)^d}{\lambda_d}} \cdot \nu \cdot (x+1)^d + z \cdot \nu \cdot (x+1)^d  
    \\&\qquad+ \nu \cdot x^{d+1} -  \nu \cdot z\cdot x^d - \frac{(x+1)^{d+1}}{d+1}
    \\&= (x+1)^{d+1}\left(\left( \frac{1}{d+1} - 1\right) \frac{\nu^{\frac{d+1}{d}}  }{\sqrt[d]{\lambda_d}} - \frac{1}{d+1}\right) - \left((\lambda_d-1)\frac{(z+1)^{d+1}}{d+1}  - z \cdot \nu \left( (x+1)^d  - x^d\right)\right)
    \\&\qquad + \nu \cdot x^{d+1}.
\end{align*}
The above expression is minimized for, \[z = \sqrt[d]{\frac{\nu((x+1)^d - x^d)}{\lambda_d - 1}} - 1.\]
Substituting for $z$ we get,
\begin{align*}
    &\ge (x+1)^{d+1}\left(\left( \frac{1}{d+1} - 1\right) \frac{\nu^{\frac{d+1}{d}}  }{\sqrt[d]{\lambda_d}} - \frac{1}{d+1}\right)  - \left( \frac{1}{d+1} - 1\right)\frac{\nu^\frac{d+1}{d}}{\sqrt[d]{\lambda_d - 1}} \left((x+1)^d - x^d\right)^\frac{d+1}{d} 
    \\&\qquad+ \nu \cdot x^{d+1} - \nu \left((x+1)^d - x^d\right)
\end{align*}
For any $\nu \ge 0$ the above expression is,
\begin{align*}
    &\ge (x+1)^{d+1}\left(\left( \frac{1}{d+1} - 1\right) \frac{\nu^{\frac{d+1}{d}}  }{\sqrt[d]{\lambda_d}} - \frac{1}{d+1}\right) + \nu \left( x^{d+1} - (x+1)^d + x^d\right)
\end{align*}
Now, what is left to show is that there exists a $\nu \ge 0$ such that, 
\begin{align*}
    (x+1)^{d+1}\left(\left( \frac{1}{d+1} - 1\right) \frac{\nu^{\frac{d+1}{d}}  }{\sqrt[d]{\lambda_d}} - \frac{1}{d+1}\right) + \nu \left( x^{d+1} - (x+1)^d + x^d\right) \ge 0
\end{align*}
equivalent to,
\begin{align*}
    \nu \left( x^{d+1} - (x+1)^d + x^d\right) \ge (x+1)^{d+1}\left(\left( 1- \frac{1}{d+1} \right) \frac{\nu^{\frac{d+1}{d}}  }{\sqrt[d]{\lambda_d}} + \frac{1}{d+1}\right)\tag{14}
\end{align*}
For $\nu = \sqrt[d+1]{\lambda_d}$ and $\forall x \in \mathbb{N}$ such that, \[\frac{(x+1)^{d+1}}{\left( x^{d+1} - (x+1)^d + x^d\right)} \le \sqrt[d+1]{\lambda_d},\] the inequality (14) holds. From the fact that, \[\lim_{x \to \infty}\frac{(x+1)^{d+1}}{\left( x^{d+1} - (x+1)^d + x^d\right)} = 1,\]

and choice of $\lambda_d = \rho_d$  (Table~\ref{table:apxPNE}), gives for $0< d \le 5$, that $x \ge 150$ is sufficient.
\end{proof}

We further note, that for a given $\lambda_d > 0$, for each $n$ and $z$, we only need to consider a limited range for $m$.
      \begin{lemma}
 \label{lem:largeMStrech}
   For fixed $n,z$, if $\lambda_d \cdot \sum_{i=z+1}^{m+z} i^{d} - m f'(n+z+1) + n f'(n+z) \ge \sum_{i=z+1}^{n+z} i^{d} $ is true $\forall m \le (n+z+1)^2(d+1)$, it also holds $\forall m > (n+z+1)^2 (d+1)$.

 \end{lemma}
  \begin{proof}For any given $n,z \in \mathbb{N}$, and $\lambda_d > 0$, we show for all $m \ge (n+z+1)^2(d+1)$ that, \[\lambda_d \sum_{i=z+1}^{m+z} i^d - m f'(n+z+1) + n f'(n+z) - \sum_{i=z+1}^{n+z} i^d \ge 0 .\]

  We first upper bound the feasible values for $f'$. From the smoothness condition in Lemma~\ref{lem:subset}, observe that for $n=z=0$ and $m=1$ the inequality implies that $f'(1) \le \lambda_d$. 
  Furthermore, by the choice of $m=n$, we get \begin{align*}
    f'(n+z+1)&\le (\lambda_d-1)\sum_{i=1}^{n+z} \left(\frac{1}{i}\sum_{j=1}^i j^d \right) + \lambda_d
    \\&\le (\lambda_d-1)\sum_{j=1}^{n+z} j^d + \lambda_d
    \\&\le (\lambda_d-1)(n+z)^{d+1} + \lambda_d
 \end{align*} 
    This gives,
    \begin{align*}
     &\lambda_d \sum_{i=z+1}^{m+z} i^d - m f'(n+z+1) + n f'(n+z) - \sum_{i=z+1}^{n+z} i^d
     \\&\ge \lambda_d \sum_{i=z+1}^{m+z} i^d - m \left( (\lambda_d-1)(n+z)^{d+1} + \lambda_d\right) + n f'(n+z) - \sum_{i=z+1}^{n+z} i^d
     \\&\ge \lambda_d \sum_{i=z+1}^{m+z} i^d - m \left( (\lambda_d-1)(n+z)^{d+1} + \lambda_d\right) - \sum_{i=z+1}^{n+z} i^d
     \\&\ge \lambda_d \sum_{i=z+1}^{m+z} i^d - \lambda_d m(n+z)^{d+1} + m(n+z)^{d+1} - m\lambda_d - (n+z)^{d+1}
     \\&= \lambda_d \sum_{i=z+1}^{m+z} i^d - \lambda_d m(n+z)^{d+1} + (m-1)(n+z)^{d+1} - m\lambda_d
     \\&\ge \lambda_d \sum_{i=z+1}^{m+z} i^d - \lambda_d m(n+z)^{d+1} - m\lambda_d
     \end{align*} 
      Now we can bound for $m\ge (d+1)(n+z+1)^2$.
     \begin{align*}
     &\lambda_d \sum_{i=z+1}^{m+z} i^d - \lambda_d m(n+z)^{d+1} - m\lambda_d
     \\&=\lambda_d \sum_{i=z+1}^{m+z} i^d - \lambda_d m((n+z)^{d+1} +1)
     \\&\ge\lambda_d \frac{m(m+z)^d}{d+1} - \lambda_d m((n+z)^{d+1} +1)
     \\&\ge\lambda_d \frac{m^{d+1}}{d+1} - \lambda_d m((n+z)^{d+1} +1)
     \\&\ge 0.
 \end{align*}
 \end{proof}

     By Lemma \ref{lem:largeNPotential} and \ref{lem:largeMStrech} it remains to solve the following linear program LP$_\phi^K$ to obtain our results $\rho_d$  as listed in Table~\ref{table:apxPNE} for $d\le 5$.
         \begin{align*}
		\min \rho_d
		\\\rho_d \cdot \sum_{i=z+1}^{m+z} i^d - m f'(n+z+1) + n f'(n+z) &\ge \sum_{i=z+1}^{n+z} i^d &\forall (n+z)\in [0,K), 
		\\&&\forall m\in [0,(K+1)^2(d+1)]
		\\f'(K) &\le \nu K^d&  
		\\f'(n) &\ge 0   &\forall n \in [0,K]
	\end{align*}
\begin{corollary}
For every congestion game with polynomial cost functions of degree $d \le 5$ and for every constant $\epsilon>0$, the algorithm computes a $(\rho_d+\epsilon)$-approximate pure Nash equilibrium in polynomial time.
\end{corollary}

\section{Load Dependent Universal Taxes}
\label{section_extension}

  We now look at an extension of Lemma~\ref{def:usmooth} for computing load dependent universal taxes in congestion games. We give a rather simple approach to locally (on resource) compute load dependent {\em universal} taxes. Table~\ref{table:locality} lists the improved PoA bounds under refundable taxation using our technique for congestion games with resource cost functions that are bounded degree polynomials of maximum degree $d$. 
By the smoothness argument (Theorem~\ref{theo:extension}\footnote{ See appendix~\ref{missing_proofs_apxalgo}.}, \cite{Roughgarden:2015:IRP:2841330.2806883}) the new bounds immediately extends to mixed, (coarse) correlated equilibria, and outcome generated by no-regret sequences. Moreover, since the linear program that computes the cost or tax function does only depend on the original cost function of that resource, the computed taxes are robust against perturbations of the instance such as adding or removing of resources or players.

    We seek to compute universal load dependent taxes that minimize the PoA under refundable taxation. We consider the following optimization problem.    
	For an objective function $h(s)= \sum_{e \in E} n_e(s)\cdot f_e(n_e(s))$, find functions $f'_e$ that satisfies Lemma~\ref{def:usmooth} minimizing $\lambda$. For a resource objective function $h_e(n_e(s)) = n_e(s)\cdot f_e(n_e(s))$ and a bound on the number of players $N$, this can be easily solved by the following linear program LP$_{\textsc{sc}}$ with the variables $f'_e(1),\ldots,f'_e(N)$, and $\lambda_e$. 
    \begin{align*}
		\min \lambda_e
		\\\lambda_e \cdot h_e(m) - m f'_e(n+1) + n f'_e(n) &\ge h_e(n) &\text{for all } n\in [0,N], m\in [0,N]
		\\f'_e(n) &\ge 0   &\text{for all } n \in [0,N]
	\end{align*}
  	Observe that, we can solve LP$_{\textsc{sc}}$ locally for each resource with cost function $f_e$. For the LP solution $\lambda_e$ and $f'_e(n)$, define the tax function as $t_e(n) := f'_e(n) - f_e(n)$. The resulting price of anarchy under taxation is then $\lambda := \max_{e \in E} \lambda_e$. 
  	
	For any (distributed) local search algorithm (such as Bjelde et al.~\cite{Bjelde:2017:BAA:3087556.3087597}) that seeks to minimize the social cost $c(s) =  \sum_{e \in E} n_e(s) f_e(n_e(s))$, we define $\zeta_\textsc{sc}(s) := \sum_{e \in E} \sum_{i = 1}^{n_e(s)} f'_e(i)$ as a pseudo-potential function. Then, from Lemma~\ref{def:usmooth} it is guaranteed that every local optimum with respect to $\zeta_\textsc{sc}(s)$ has an approximation factor of at most $\lambda := \max_{e \in E} \lambda_e$ with respect to the social cost $c(s)$. 
		Using approximate local search by Orlin et al.~\cite{Orlin:2004:ALS:982792.982880}, we can compute a solution close to that in polynomial time and more so to state the following.
	\begin{corollary}
    For every congestion game 
    the $\epsilon$-local search algorithm using $\zeta_\textsc{sc}(s)$, produces a $\lambda(1 + \epsilon)$ local optimum in running time polynomial in the input length, and $1/\epsilon$.
	\end{corollary}
\noindent \textbf{Linear and Polynomial Cost Functions}
    
For the interesting case of polynomial resource cost functions of maximum degree $d$, similar to Section~\ref{CG_stretch}, we show that for polynomials of small degree, it is sufficient to restrict the attention to the first $K=1154$ values of the cost functions. Hence, we only need solve a linear program of constant size. The following lemma states that for the values of $n$ greater than $K$ and an appropriate value of $\nu$, and $\lambda_d$, we can easily obtain $(\lambda_d,0)$-smoothness by choosing $f^\prime(n)=\nu n^d$.
    
\begin{lemma}
    \label{lem:largeKpoa}
    For $d \le 5$ and $n\ge  1154$, the function $f'(n)= \nu n^d$ with $\nu = \sqrt[d+1]{(d+1)\lambda_d}$ is $(\lambda_d,0)$-smooth with respect to $h(n)=n^{d+1}$ and an appropriate $\lambda_d$. 
     \end{lemma}
      \begin{proof}For each degree $d$ and an appropriate value of $\lambda_d = \Psi_d$ (Table~\ref{table:locality}), we need to show that there exist a $\nu$ such that the the smoothness condition in Lemma~\ref{def:usmooth} holds for all $n\ge 1154$, i.e.,
\[\lambda_d m^{d+1} - m \nu (n+1)^d + n \nu n^d - n^{d+1} \ge 0. \]
For all $\lambda_d > 0$,
\begin{align*}
    &\lambda_d m^{d+1} - m \nu (n+1)^d + n \nu n^d - n^{d+1}
    \\&=\lambda_d m^{d+1} - m \nu (n+1)^d + n^{d+1}(\nu-1).
\end{align*}
The above expression is minimized at, \[ m = \sqrt[d]{\frac{\nu(n+1)^d}{(d+1)\lambda_d}}.\]
Substituting for $m$ gives,
\begin{align*}
    &=\lambda_d \left(\sqrt[d]{\frac{\nu(n+1)^d}{(d+1)\lambda_d}}\right)^{d+1} - \left(\sqrt[d]{\frac{\nu(n+1)^d}{(d+1)\lambda_d}}\right) \nu (n+1)^d + n^{d+1}(\nu-1)
    \\&=\left(\frac{1}{d+1}\right) \frac{\nu^\frac{d+1}{d}}{\sqrt[d]{(d+1)\lambda_d}} (n+1)^{d+1}- \frac{\nu^\frac{d+1}{d}}{\sqrt[d]{(d+1)\lambda_d}} (n+1)^{d+1} + n^{d+1}(\nu-1)
    \\&=\left(\frac{1}{d+1} - 1\right)\frac{\nu^\frac{d+1}{d}}{\sqrt[d]{(d+1)\lambda_d}} (n+1)^{d+1} + n^{d+1}(\nu-1)
\end{align*}
We need to show that there exists a $\nu \ge 0$ such that,
\begin{align*}
    &\left(\frac{1}{d+1} - 1\right)\frac{\nu^\frac{d+1}{d}}{\sqrt[d]{(d+1)\lambda_d}} (n+1)^{d+1} + n^{d+1}(\nu-1) \ge 0
\end{align*}
equivalent to,
\begin{align*}
    & n^{d+1}(\nu-1) \ge \left(1- \frac{1}{d+1}\right)\frac{\nu^\frac{d+1}{d}}{\sqrt[d]{(d+1)\lambda_d}} (n+1)^{d+1}\tag{15}
\end{align*}
For $\nu = \sqrt[d+1]{(d+1)\lambda_d}$ and $\forall n \in \mathbb{N}$ such that,
\[ \frac{(n+1)^{d+1}}{n^{d+1}}\left( 1- \frac{1}{d+1}\right) \le \left(\sqrt[d+1]{(d+1)\lambda_d} -1 \right),\]
the inequality (15) holds. Also, using the fact that,
\[ \lim_{n \to \infty} \frac{(n+1)^{d+1}}{n^{d+1}} = 1,\]
and choice of $\lambda_d = \Psi_d$ (Table~\ref{table:locality}), gives for $0 < d \le 5$, that $n\ge 1154$ is sufficient.
\end{proof}

We further note, that for a fixed $\lambda_d > 0$ and for each $n$, we only need to consider a limited range for $m$ in the LP$_{\textsc{sc}}$.
     \begin{lemma}
    \label{largeMPoA}
   For a fixed $n$, if $\lambda_d \cdot m^{d+1} - m f(n+1) + n f(n) \ge n^{d+1}$ is true for all $m \le (n+1)^2$, it also holds for all $m > (n+1)^2$.
    \end{lemma}
    \begin{proof}For any given $n\in \mathbb{N}$ and $\lambda_d > 0$, we show for all $m \ge (n+1)^2$ that, \[\lambda_d m^{d+1} - m f'(n+1) + n f'(n) - \sum_{i=1}^{n} n^{d+1} \ge 0 .\]

    We first upper bound the feasible values for $f'(n)$. From the smoothness condition in Lemma~\ref{def:usmooth}, note that for $n=0$ and $m=1$, implies that $f'(1) \le \lambda_d$. Furthermore, by the choice of $m=n$, we get $f'(n+1) \le f'(n) + \frac{\lambda_d -1}{n} n^{d+1}$. By recursion, we obtain $f'(n+1)\le (\lambda_d-1)\sum_{i=1}^{n} i^d + \lambda_d$
    which we can simply bound by
    $f'(n+1) \le (\lambda_d - 1) \cdot n^{d+1} + \lambda_d$.
     Now we can bound for $m\ge (n+1)^2$.
     
     \begin{align*}
     &\lambda_d m^{d+1} - m f'(n+1) + n f'(n) - n^{d+1}
     \\&\ge \lambda_d m^{d+1} - m \left((\lambda_d -1 ) n^{d+1} + \lambda_d\right)- n^{d+1}
     \\&= \lambda_d m^{d+1} - m\lambda_d n^{d+1} + mn^{d+1} - m\lambda_d - n^{d+1}
     \\&\ge \lambda_d m^{d+1} - m \lambda_d  n^{d+1} - m\lambda_d
     \\&\ge 0
     \end{align*}
    \end{proof}

   As a consequence of Lemma \ref{lem:largeKpoa} and \ref{largeMPoA} it only remains to solve the following linear program of constant size for each $d \le 5$ to obtain our results $\Psi_d$ (listed in Table~\ref{table:locality}). Our results match the recent results that were obtained independently by Paccagnan et al.~\cite{chandan1}.
    \begin{align*}
		\min \Psi_d
		\\\Psi_d m^{d+1} - m f'(n+1) + n f'(n) &\ge n^{d+1} &\forall n\in [0,K), m\in [0,(K+1)^2]
		\\f'(K) &\le \nu K^d&  
		\\f'(n) &\ge 0   &\forall n \in [0,K]
	\end{align*}

	\begin{corollary}\label{cor_ls}
	For every congestion game with polynomial cost functions of degree $d \le 5$, each cost function $f'_e$ can be computed in constant time and the resulting game is $(\Psi_d,0)$-smooth with respect to social cost.
	\end{corollary}
	
	  	\subsubsection*{Lower Bound}
   	\label{LB_Section}
	
	Any feasible solution to the linear program LP$_h$ emerging from Lemma~\ref{def:usmooth} are cost functions $f'_e:\mathbb{N} \mapsto R_+$ that guarantees that the objective value associated with the function $h$ is at most $\lambda := \max_{e \in E} \lambda_e$. We can show that this is in fact optimal. That is, LP$_h$ is not only optimizing the smoothness inequality, but also that there exists no other resource cost function that can guarantee a smaller objective value than $\lambda$. To that end, we consider the dual of  LP$_h$ and show that for every feasible solution of the dual, we can construct an instance of a selfish scheduling game with an objective value that is equal to the value of the dual LP solution, regardless of the actual cost function of the game. 
	
	We construct a selfish scheduling game with identical machines. This is a congestion game in which players' strategies are singletons. Furthermore, each resource has the same cost function and hence, players only seek to choose a resource with minimal load. Obviously, every equilibrium in the scheduling game is an equilibrium in a congestion game in which the resource cost function is an arbitrary non-decreasing function.
	The dual program LPD$_h$ is as follows,
		\begin{align*}
		\max \sum_{n=0}^N \sum_{m=0}^N h(n) \cdot y_{n,m}\tag{16}
		\\\sum_{n=0}^N \sum_{m=0}^N h(m) \cdot y_{n,m}&\le 1\qquad\ \tag{17}
		\\\sum_{m = 0}^N n \cdot y_{n,m}  - \sum_{m=0}^N m \cdot y_{n-1,m} &\le 0 \hspace{1cm}  &\text{for all }  n \in [0, N]
		\\y_{n,m} &\geq 0 \hspace{1cm}  &\text{for all } n,m \in [0, N]
		\end{align*}

\begin{lemma}
\label{lowerbound_lemma}
		Every optimal solution of LPD$_h$ with objective value $\lambda$ can be turned into an instance of selfish scheduling on identical machines with an 
		objective value of $\lambda - \epsilon$ for an arbitrary $\epsilon > 0$.
	\end{lemma}
		\begin{proof}Let $\tilde{y}$ be a feasible solution to LPD$_h$ with $\lambda = \sum_{n=0}^N \sum_{m=0}^N h(n) \cdot \tilde{y}_{n,m}$. 
	We round down each $\tilde{y}_{n,m}$ to rational numbers $y_{n,m} \ge (1-\frac{\epsilon}{\lambda}) \tilde{y}_{n,m}$ and let $M$ be a sufficiently large scaling factor such that each $y_{n,m}\cdot M$ is an integer. We construct a congestion game as follows. 
	The game $G = (\mathcal{N}, R, \{S_u\}_{u \in \mathcal{N}}, \{c_r\}_{r \in R})$ consist of a set of players $\mathcal{N} =\bigcup_{n \in [0,N]} \mathcal{N}_{n}$, where each set  $\mathcal{N}_{n}$ consists of $n\cdot \sum_{m =0}^N y_{n,m} \cdot M $ many players. The  set of resources is $R = \bigcup_{n \in [0,N], m \in [0,N]} R_{n,m}$, where $R_{n,m}$ represents a pool consisting of $y_{n,m} \cdot M$ identical machines.  
	
	We will make sure that in the game $G$, there is an equilibrium $s^*$ in which on each machine in each set $R_{n,m}$ there are exactly $n$ many players. Hence, using (16) the total cost over all resources in $s^*$ is $\sum_{n=0}^N \sum_{m=0}^N h(n) \cdot y_{n,m}\cdot M \ge \lambda (1-\epsilon) M$. Additionally, there is a state $s$ in which  on each machine in each set $R_{n,m}$ there are exactly $m$ many players. Hence, using (17) the total cost over all resources in $s$ is $\sum_{n=0}^N \sum_{m=0}^N h(m) \cdot y_{n,m} \cdot M \le M$.  
	
	Each player from a set $\mathcal{N}_{n}$ in the game has two strategies, which we call an equilibrium strategy and an optimal strategy.
	The equilibrium strategy consists of one particular resource $r \in \bigcup_{m \in {[0,N]}} R_{n,m}$
	and the optimal strategy of a particular resource $r \in \bigcup_{m \in {[0,N]}} R_{n-1, m}$. The assignment of resources to strategies is such that each resource $r \in R_{n, m}$ belongs to exactly $n$ equilibrium strategies of players from $\mathcal{N}_{n}$ and at most $m$ optimal strategies of players from $\mathcal{N}_{n+1}$. 
	Note that the existence of such an assignment is guaranteed by the feasibility of $y$ and hence, $\sum_{m = 0}^N n \cdot y_{n,m}  \le \sum_{m=0}^N m \cdot y_{n-1,m}$. If each player chooses the equilibrium strategy, we obtain a state $s^*$ as described above which is a pure Nash equilibrium as switching to the optimal strategy yields exactly the same load. If each player chooses the optimal strategy we obtain a state $s$ as described above. Therefore, the objective value is at least $\lambda- \epsilon$.
\end{proof}

	From Lemma~\ref{lowerbound_lemma} we remark that the taxes computed by LP$_{\textsc{sc}}$ are optimal. Evidently our lower bound of $2.012$ for congestion games with linear cost functions matches the price of anarchy bound for selfish scheduling games on identical machines~\cite{CaragiannisFKKM06}.
  
		\section{Conclusion and Open Problems}
The most interesting question which was the initial motivation for this work is the complexity of approximate equilibria. We find it very surprising that the technique yields such a significant improvement, e.g., for linear congestion games from $2$ to $1.61$, by using essentially the same algorithm of Caragiannis et al.~\cite{Caragiannis2011}. 
 
	    However, the algorithmic technique is limited only by the lower bound for approximation factor of the stretch implied in Roughgarden~\cite{Roughgarden:2014:BNE:2706700.2707451}. Hence, further significant improvements may need new algorithmic ideas. 
	    On the lower bound side, not much is known for linear or polynomial congestion games. The only computational lower bound for approximate equilibria is from Skopalik and V{\"o}cking~\cite{Skopalik:2008:IPN:1374376.1374428} using unnatural and very steep cost functions.
	    
	    We believe that the technique of perturbing the instance of an (optimization) problem such that a simple local search heuristic (or an equilibrium) guarantees an improved approximation ratio can be applied in other settings as well. It would be interesting to see, whether one can achieve similar results for variants and generalizations of congestion games such as weighted~\cite{Awerbuch:2005:PRU:1060590.1060599}, atomic- or integer-splittable~\cite{doi:10.1002/net.3230030104, Roughgarden:2011:LSP:2133036.2133058} congestion games, scheduling games~\cite{correa2012efficiency, gairing2010computing, cole2015decentralized}, etc.
	    Considering other heuristics such as greedy or one-round walks~\cite{CHRISTODOULOU201213, Bilo2011oneround, KlimmGreedy, bil_et_al} would be another natural direction. 

\bibliographystyle{plain}
\bibliography{references}
\appendix
    \section{Appendix}
    \label{missing_proofs_apxalgo}
    \subsection*{Missing proofs of  Section~\ref{def_prelim}}

	\begin{prooff}{Theorem~\ref{theo:smoothness}}
 
		Let $s$ be an arbitrary pure Nash equilibrium and $s^*$ be the optimal solution. From the Nash inequality we know,
		\[ \forall u \in \mathcal{N},~c_u(s)  \le c_u(s_u^*, s_{-u}).\]
		Then, summing over all the $\mathcal{N}$ players gives,
		\[ \sum_{u \in \mathcal{N}}c_u(s)  - \sum_{u \in \mathcal{N}}c_u(s_u^*, s_{-u}) \le 0. \]
		
		By the definition of $(\lambda, \mu)$ smoothness in Definition~\ref{def:lambdasmooth} we  have,
		\[ (1-\mu)\cdot h(s)  \le \lambda \cdot  h(s^*) + \sum_{u \in \mathcal{N}}c_u(s) -  \sum_{u \in \mathcal{N}}c_u(s_u^*, s_{-u})\]
		From the Nash inequality the theorem follows.
	\end{prooff}

		\begin{prooff}{Lemma~\ref{def:usmooth}}
		
		Let $s$ and $s^*$ be arbitrary solutions. Summing the inequality of the lemma with $m=n_e(s^*)$ and $n=n_e(s)$ for all $e \in E$ gives,
		\begin{align*}
		\lambda \sum_{e \in E} h_e(n_e(s^*) &\ge  \sum_{e \in E} n_e(s^*) f'_e(n_e(s)+1) - \sum_{e \in E}  n_e(s) f'_e(n_e(s))\  +  \sum_{e \in E} h_e(n_e(s))\\
			 \lambda \cdot  h(s^*) &\geq   \sum_{u \in \mathcal{N}}  c'_u(s_u^*,s_{-u}) - \sum_{u \in \mathcal{N}}  c'_u(s) + h(s)
		\end{align*}
		which is the $(\lambda, 0)$-smoothness condition of Definition~\ref{def:lambdasmooth}.
 \end{prooff}
 
 \begin{prooff}{Lemma~\ref{lem:subset}}
 
 The proof is analogous to the proof of Lemma~\ref{def:usmooth}. Consider an arbitrary subset of players $F \subseteq \mathcal{N}$, an arbitrary state $s$, and a subgame $G_s^F:=(F,E,(S_u)_{u \in F}, (f^F_e)_{e_\in E})$ induced by freezing the remaining players from $\mathcal{N} \setminus F$, that is, let $f^F_e(x):=f_e(x+n_e^{\mathcal{N}\setminus F}(s))$ where $n_e^{\mathcal{N}\setminus F}(s)$ is the number of players outside of $F$ on resource $e$ in the state $s$. Let $s^*$ be an arbitrary solution. Summing the inequality of the lemma with $m=n^F_e(s^*)$ and $n=n^F_e(s)$ for all $e \in E$ gives,
 		\begin{align*}
		\lambda \sum_{e \in E} \phi_e^F(n_e(s^*) &\ge  \sum_{e \in E} n_e^F(s^*) f'_e(n_e(s)+1) - \sum_{e \in E}  n_e^F(s) f'_e(n_e(s)) +  \sum_{e \in E} \phi^F_e(n_e(s))
		\end{align*}
		equivalent to,
		\begin{align*}
            \lambda \cdot \phi^F(s^*) &\ge \sum_{u \in F} c'_u(s_u^*,s_{-u}) - \sum_{u \in F} c'_u(s)+\phi^F(s),
		\end{align*}
		which is the $(\lambda, 0)$-smoothness condition of Definition~\ref{def:usmooth_subset}.
 \end{prooff}
	\begin{theorem}[Extension Theorem]
	\label{theo:extension}
		For every $(\lambda, \mu)$-smooth cost-minimization game $G$ with respect to an arbitrary objective function $h$, every coarse correlated equilibrium $\sigma$, and every outcome $s^*$, \[\mathbb{E}_{s \sim \sigma}[h(s)] \le \frac{\lambda}{1 - \mu} \cdot h(s^*).\]
	\end{theorem}
	\begin{proof}
        The proof is analogous to Roughgarden's  proof~\cite{Roughgarden:2015:IRP:2841330.2806883}. From the $(\lambda, \mu)$-smoothness condition in Definition~\ref{def:lambdasmooth}, we have,
		\begin{align*}
			\mathbb{E}_{s \sim \sigma}[h(s)]
			&\le \frac{1}{1-\mu}\mathbb{E}_{s \sim \sigma}\Bigg[\lambda\cdot h(s^*) + \sum_{u \in \mathcal{N}} c_u(s) - \sum_{u \in \mathcal{N}} c_u(s_u^*, s_{-u})\Bigg]
			\\&=  \frac{1}{1 - \mu}\Bigg[\lambda\cdot h(s^*) + \sum_{u \in \mathcal{N}} \mathbb{E}_{s \sim
				\sigma}[c_u(s)] - \sum_{u \in \mathcal{N}} \mathbb{E}_{s \sim \sigma}[c_u(s_u^*, s_{-u})]\Bigg]
			\\\text{From the definition of CCE,}
			\\&\le\frac{\lambda}{1 - \mu}\cdot h(s^*).
		\end{align*}    
		Hence, the theorem.
	\end{proof}
\end{document}